%% file: paper.tex
\documentclass{article}

\input{content/preamble}

\usepackage{fullpage}
\usepackage[hidelinks]{hyperref}
\newcommand*{\email}[1]{\href{mailto:#1}{\nolinkurl{#1}} } 

\title{Optimal Strategies in Weighted Limit Games\thanks{We would like to thank Alexander Weinert for fruitful discussions leading to this work and the reviewers for their detailed feedback, which considerably improved the paper.}}
\author{
\textbf{Aniello Murano}\\
Università degli Studi di Napoli “Federico II”, Naples, Italy\\
\email{murano@na.infn.it}\bigskip\\
\textbf{Sasha Rubin}\\
The University of Sydney, Sydney, Australia\\
\email{sasha.rubin@sydney.edu.au}\bigskip\\
\textbf{Martin Zimmermann}\\
University of Liverpool, Liverpool, UK\\
\email{martin.zimmermann@liverpool.ac.uk}
}
\date{\vspace{-4ex}}

\begin{document}
\maketitle

\input{content/abstract}	


\section{Introduction}
\label{section-intro}
\input{content/intro}

\section{Definitions}
\label{section-defs}
\input{content/defs}

\section{Weighted Limit Games}
\label{section-limit}
\input{content/limit}


\subsection{Computing Optimal Strategies in Weighted Reachability Games}
\label{subsection-reachalgo}
\input{content/reachalgo}


\subsection{Computing Optimal Strategies in Weighted Limit Games}
\label{subsection-limitalgo}
\input{content/limitalgo}

\section{Limit Games in Infinite Arenas}
\label{section-pushdown}
\input{content/furtherwork}

\section{Related Work}
\label{section-relatedwork}
\input{content/relatedwork}

\section{Conclusion}
\label{section-conc}
\input{content/conc}


\bibliographystyle{plain}
\bibliography{content/biblio}


\clearpage
\appendix
\section{Appendix}
\input{content/appendix}

\end{document}

%% file: content/preamble.tex
\usepackage[utf8]{inputenc}
\usepackage[T1]{fontenc}
\usepackage{mathtools}
\usepackage{tikz}
\usetikzlibrary{automata}
\usepackage{amsmath}
\usepackage{amssymb}
\usepackage[]{todonotes}

\usepackage{amsthm}

\theoremstyle{theorem}
\newtheorem{lemma}{Lemma}

\newtheorem{theorem}{Theorem}
\newtheorem{remark}{Remark}
\newtheorem{corollary}{Corollary}
\newtheorem{example}{Example}

\usetikzlibrary{
  calc,
  arrows,
  positioning,
  shapes,
  shadows,
  automata,
  external,
  shadows,
  decorations.pathreplacing,
  decorations.pathmorphing,
  decorations.markings,
  backgrounds,
  patterns
}

\tikzset{
  node distance=6em,
  >=stealth',
  shadow/.style		= {opacity=.25, black, shadow xshift=0.08, shadow yshift=-0.08},
  shadowt/.style		= {opacity=.25, black, shadow xshift=0.05, shadow yshift=-0.05},
  parshadow/.style	= {opacity=.25, black, shadow xshift=0.04, shadow yshift=-0.04},
  plainnode/.style 	= {draw, ultra thick, fill=gray!10},
  pl0/.style			= {circle, minimum size=9mm, plainnode, drop shadow = {shadow}},
  pl1/.style			= {rounded corners = 1, minimum size=9mm, plainnode, drop shadow = {shadow}},
  pl0s/.style			= {circle, minimum size=9mm, plainnode, drop shadow = {shadowt}, scale=.5},
  pl1s/.style			= {rounded corners = 1, minimum size=9mm, plainnode, drop shadow = {shadowt}, scale=.5},
  parpl0/.style		= {ellipse , text centered, minimum size=9mm,inner sep=-1pt, plainnode, drop shadow = {parshadow}},
  parpl1/.style		= {rectangle, minimum size=0mm, inner sep = 4pt,minimum size=9mm, text centered, plainnode, drop shadow = {parshadow}},
  ac/.style			= {double},
  w0/.style			= {fill=blue!30},
  w1/.style			= {fill=red!30},
  r/.style			= {bend left=18},
  l/.style			= {bend right=18},
  every picture/.style=thick
}

\newcommand{\myquot}[1]{``#1''}

\renewcommand{\epsilon}{\varepsilon}
\newcommand{\bigo}[0]{\mathcal{O}}

\newcommand{\card}[1]{|#1|}
\newcommand{\size}[1]{|#1|}
\newcommand{\set}[1]{\{ #1 \}}
\newcommand{\nats}[0]{\mathbb{N}}

\newcommand{\initstate}[0]{I}

\newcommand{\aut}{\mathfrak{A}}

\newcommand{\arena}{\ensuremath{\mathcal{A}}}
\newcommand{\game}{\ensuremath{\mathcal{G}}}
\newcommand{\win}{\ensuremath{\mathrm{Win}}}
\newcommand{\weight}[0]{w}

\newcommand{\val}{\text{val}}
\newcommand{\ranking}{r}
\newcommand{\rankings}{\mathcal{R}}
\newcommand{\natsclosed}{\overline{\nats}}
\newcommand{\lift}{\ell}
\newcommand{\col}{c}

\newcommand{\mem}{\mathcal{M}}
\newcommand{\init}{\text{init}}
\newcommand{\update}{\mathrm{upd}}
\newcommand{\nxt}{\mathrm{Nxt}}
\newcommand{\ext}{\mathrm{ext}}

\newcommand{\complete}{\mathrm{cmplt}}
\newcommand{\limitlift}{\ell_L}

%% file: content/abstract.tex
\begin{abstract}
We prove the existence and computability of optimal strategies in weighted limit games, zero-sum infinite-duration games with a Büchi-style winning condition requiring to produce infinitely many play prefixes that satisfy a given regular specification. Quality of plays is measured in the maximal weight of infixes between successive play prefixes that satisfy the specification.
\end{abstract}

%% file: content/intro.tex
Reactive synthesis is an ambitious approach to the problem of producing correct controllers for reactive systems, e.g., systems continuously interacting with their environment over an infinite time horizon. 
Instead of an engineer coding the controller and then checking it for correctness against a formal specification, one automatically computes a correct-by-construction controller from the specification.

The basic case of the problem, formalized as Church's problem~\cite{Church63}, has been solved by the seminal Büchi-Landweber Theorem~\cite{BuechiLandweber69}. Here, the problem is recast as a game-theoretic one: 
Given a finite graph describing the interaction between the desired controller and its environment, and a winning condition representing the controller's specification, determine whether the \myquot{controller player} has a winning strategy for this game. 
If yes, Büchi and Landweber proved that she has a finite-state winning strategy, e.g., one that can be implemented by a finite automaton with output. 
Such a strategy can be seen as a controller that satisfies the specification. We refer to these lecture notes~\cite{F16} and the references therein for a contemporary overview of reactive synthesis.

Ever since the seminal work of Büchi and Landweber, their result has been extended in various directions, e.g., more expressive winning conditions, infinite state spaces, stochastic settings, settings with imperfect information, etc. 
All these are motivated by the quest to model ever more aspects of relevant application domains.

Recently, another aspect has received considerable attention: 
Oftentimes, specifications are qualitative but some controllers are more desirable than others. 
Consider, for example, a controller that has to bring a system into a desirable state. Then, it is often desired, although not formally specified, that the state is reached as quickly as possible or with the minimal amount of resource consumption.
Much effort has been put into computing controllers that satisfy such \myquot{nonfunctional} requirements. 

But not every specification is a reachability property. As another example, assume we need to generate an arbiter that controls access to some shared resource. 
A typical specification here is to require that every request to the resource is eventually granted~\cite{WallmeierHT03}.
Again, we typically prefer controllers that grant requests as quickly as possible. 
Note that this specification is not a simple reachability property that requires to reach a certain set of states, but a recurrence property that requires to infinitely often reach a state in which no request is pending. 
The optimization criterion then asks to minimize the maximal time between visits to such states. 

Formally, recurrence properties are captured by Büchi games (see, e.g.,~\cite{GraedelThomasWilke02}), i.e., games whose underlying graphs come with a set of desirable vertices that need to be visited infinitely often for the controller player to win. 
In this work, we consider a slightly different approach:
We equip the graph describing the interaction with labels on vertices (think of the set of atomic propositions holding true in this state) and the edges with nonnegative weights (capturing the cost or time it takes to make this transition). 
The winning condition is induced by a deterministic finite automaton processing finite label sequences and is satisfied by an infinite play if it has infinitely many prefixes whose label sequence is accepted by the automaton.\footnote{It is not hard to reduce this setting to the one of classical Büchi games by taking the product of the graph and the automaton.}
Now, the quality of a play is measured as the maximal weight of an infix between two successive prefixes whose label sequences are accepted by the automaton.
Finally, the quality of a strategy is obtained by maximizing over the values of the plays that are consistent with~it. 

By separating the graph modeling the interaction and the specification automaton, we obtain a fine-grained analysis of the complexity of computing controllers and the complexity of implementing controllers (measured in their number of states).
In detail, our contributions are as follows:
\begin{enumerate}
	\item\label{item} We show that every such game has an optimal strategy for the controller player. 
	To prove the strategy optimal, we also show that the player representing the environment always has an optimal strategy as well, i.e., a strategy that maximizes the weight between prefixes that have a label sequence that is accepted by the automaton.
	Both strategies are obtained by a nested fixed-point characterization that generalizes the classical algorithm for solving Büchi games (see, e.g.,~\cite{CHP08}). 
	The inner fixed point is a characterization of optimal strategies in reachability games, which we use as blackbox in the outer fixed point characterization for recurrence conditions. 
			
	\item The fixed point (and the optimal strategies) can be computed in time~$\bigo(\size{V}^3 \cdot \size{E} \cdot \size{Q}^2 \cdot \size{F}^2 ) $, where $(V,E)$ is the underlying graph and $Q$ and $F$ are the sets of states and accepting states of the automaton. Here, we use the unit-cost model for arithmetic operations.
	
	\item The size of optimal strategies is bounded by $\size{V}\cdot\size{Q}\cdot\size{F}$ which is tight up to a factor of $\size{F}$. 
	
	\item The value of an optimal strategy is bounded by~$(\size{V} \cdot \size{Q} +1)\cdot W$, if it is finite at all, where $W$ is the largest weight appearing in the graph. This upper bound is shown to be tight.
	
	\item Finally, we briefly consider the case of infinite state systems. In finite graphs, if there is any controller, then there is also one with finite value. We give a very simple infinite graph in which this is no longer the case: There is a controller, but none of finite value.

\end{enumerate}

Let us stress that the results for reachability games mentioned in Item~\ref{item}) are not novel and follow from stronger results (see, e.g.,~\cite{DBLP:conf/concur/BrihayeGHM15,DBLP:journals/mst/KhachiyanBBEGRZ08}). However, we were unable to locate a reference for all the properties we require of our blackbox. Hence, for the sake of completeness, we present the construction for reachability as well, which also serves as a gentle introduction to the machinery necessary for recurrence.


%% file: content/defs.tex
Let $\nats$ denote the nonnegative integers and define $\natsclosed = \nats \cup \set{\infty}$ with $n < \infty$ and $n + \infty = \infty$ for every $n \in \nats$.
Given a finite directed graph~$(V, E)$ and $v \in V$, let $vE = \set{v' \in V \mid (v,v') \in E}$ denote the set of successors of a vertex~$v$.

\paragraph*{Finite Automata}

A deterministic finite automaton (DFA) $\aut = (Q, C, q_\initstate, \delta, F)$ consists of a finite set~$Q$ of states containing the initial state~$q_\initstate \in Q$ and the accepting states~$F \subseteq Q$, a finite set~$C$ of colors which we use as input letters, and a transition function~$\delta \colon Q \times C \rightarrow Q$. Let $\delta^*(w)$ denote the unique state that is reached by processing $w \in C^*$, i.e., $\delta^*(\epsilon) = q_\initstate$ for the empty word~$\epsilon$ and $\delta^*(w_0 \cdots w_j w_{j+1}) = \delta( \delta^*(w_0 \cdots w_j) , w_{j+1})$ for a nonempty word~$w_0 \cdots w_jw_{j+1} \in C^+$. The language of $\aut$ is $L(\aut) = \set{w \in C^* \mid \delta^*(w) \in F}$. The size of $\aut$ is defined as $\size{\aut} = \size{Q}$.

\paragraph*{Infinite Games}

Let us fix a finite nonempty set~$C$ of colors. A (weighted and colored) arena~$\arena = (V, V_0, V_1, E, \weight, \col)$ consists of a finite directed graph~$(V, E)$ whose vertices are partitioned into the vertices~$V_0$ of Player~$0$ (drawn as circles) and the vertices~$V_1$ of Player~$1$ (drawn as rectangles), a weight function~$\weight \colon E \rightarrow \nats$ (drawn as edge labels), and a coloring~$\col\colon V \rightarrow C$ (drawn as vertex labels). We require every vertex to have an outgoing edge.
A game~$\game = (\arena, \win)$ consists of an arena~$\arena$ and a (qualitative) winning condition~$\win \subseteq C^\omega$. 

A play in $\game$ is an infinite path~$\rho = v_0v_1v_2 \cdots \in V^\omega$ through $(V,E)$. 
We lift the weight function to plays and play prefixes by adding up the weights of the edges of the play (prefix). Similarly, we lift the coloring to plays and play prefixes by applying it vertex-wise. A play~$\rho$ is winning for Player~$0$ in $\game$, if $\col(\rho) \in \win$; otherwise, it is winning for Player~$1$.
 
A strategy for Player~$i \in \set{0,1}$ is a map~$\sigma \colon V^*V_i \rightarrow V$ satisfying $(v_j,\sigma(v_0\cdots v_j)) \in E$ for every $v_0 \cdots v_j \in V^*V_i$. 
A strategy~$\sigma$ for Player~$i$ is positional, if we have $\sigma(wv) = \sigma(v)$ for every $w \in V^*$ and every $v \in V_i$. We denote such strategies w.l.o.g.\ as mappings from $V_i$ to $V$.

A play~$v_0 v_1 v_2 \cdots$ is consistent with a strategy~$\sigma$ for Player~$i$, if $v_{j+1} = \sigma(v_0 \cdots v_j)$ for every $j$ with $v_j \in V_i$. 
A strategy for Player~$i$ is winning from a vertex~$v$ if every play that starts in $v$ and is consistent with the strategy is winning for Player~$i$.

\paragraph*{Memory Structures and Finite-state Strategies}

A {memory structure}~$\mem = (M, \init, \update)$ for an arena $(V, V_0, V_1, E, \weight, \col)$ consists of a finite set~$M$ of memory states, an initialization function $\init\colon V \rightarrow M$, and an update function~$\update\colon M \times V \rightarrow M$.
The update function can be extended to finite play prefixes in the usual way: $\update^*(v) = \init(v)$ and $\update^*(w v) = \update(\update^*(w ), v)$ for $w \in V^*$ and $v \in V$.
A next-move function $\nxt \colon V_i \times M \rightarrow V$ for Player~$i$ has to satisfy $(v, \nxt(v, m)) \in E$ for all $v \in V_i$ and $m \in M$.
It induces a strategy~$\sigma$ for Player~$i$ with memory~$\mem$ via $\sigma(v_0\cdots v_j) = \nxt(v_j, \update^*(v_0 \cdots v_j))$.
A strategy is called {finite-state} if it can be implemented by a memory structure.
We define $\card{\mem} = \card{M}$.
Slightly abusively, we say that the size of a finite-state strategy is the size of a memory structure implementing it.

An arena $\arena = (V, V_0, V_1, E, \weight, \col)$ and a memory structure $\mem = (M, \init, \update)$ for $\arena$ induce the expanded arena\label{def-autmem} $\arena\times\mem = (V \times M, V_0 \times M, V_1 \times M, E', \weight', \col')$ where~$E'$ is defined via $((v,m), (v',m')) \in E'$ if and only if $(v,v') \in E$ and $\update(m, v' ) = m'$. Furthermore, $\weight'((v,m),(v',m')) = \weight(v,v')$ and $\col'(v,m) = \col(v)$.
Every play $\rho = v_0 v_1 v_2\cdots$ in $\arena$ has a unique extended play $\ext(\rho) = (v_0, m_0) (v_1, m_1)
(v_2, m_2) \cdots$ in $\arena \times \mem$ defined by $m_0 = \init(v_0)$ and $m_{j+1} = \update(m_j, v_{j+1})$, i.e., $m_j = \update^*(v_0 \cdots v_j)$. The extended play of a finite play prefix in $\arena$ is defined analogously. Note that a play (prefix) and its extension have the same weight and the same color sequence.

Given a positional strategy $\sigma'$ for Player~$i$ in $\arena \times \mem$, define the finite-state strategy~$\sigma$ for Player~$i$ in $\arena$ by specifying the next-move function~$\nxt_{\sigma'}$ with $\nxt(v,m) = v'$, where $v' \in V$ is the unique vertex with $\sigma'(v,m) = (v',m')$ for some $m' \in M$.

\begin{remark}
\label{remark-extendingplayspreservesconsistency}
Let $\sigma$ and $\sigma'$ be as above and let $\rho$ a play in $\arena$. Then, $\rho$ is consistent with $\sigma$ if and only if $\ext(\rho)$ is consistent with $\sigma'$.
\end{remark}

Now let $\mem = (M, \init, \update)$ be a memory structure for the arena~$\arena = (V, V_0, V_1, E, \weight,\col)$ and let $\sigma'$ be a finite-state strategy for Player~$i$ in $\arena \times \mem = (V', V_0, V_1', E',\weight', \col')$ implemented by $\mem' = (M', \init', \update')$ and $\nxt'$. We define the product of $\mem$ and $\mem'$ as
$\mem \times \mem'  = (M \times M', \init'', \update'')$ where $\init''(v) = (\init(v), \init'(v, \init(v)))$ and 
\[\update''((m,m'),v) = (\update(m,v),  \update'(m', (v,\update(m,v)))   ) ,\]
which is a memory structure for $\arena$. Further, we obtain a finite-state strategy~$\sigma$ for Player~$i$ in $\arena$ implemented by $\mem \times \mem'$ and $\nxt$, which is defined as $\nxt(v, (m,m')) = \nxt'((v,m),m')$.
 
\begin{remark}
\label{remark-extendingplayspreservesconsistency-onsteroids}
Let $\sigma$ and $\sigma'$ be as above and let $\rho$ a play in $\arena$. Then, $\rho$ is consistent with $\sigma$ if and only if $\ext(\rho)$ is consistent with $\sigma'$, where $\ext(\rho)$ is defined with respect to $\mem$.
\end{remark}

Let $\arena$ be an arena with vertex set~$V$ and coloring~$\col \colon V \rightarrow C$, and let $\aut = (Q, C, q_\initstate, \delta, F)$ be a DFA over $C$. Then, we define $\mem_\aut = (Q, \init_\aut, \update_\aut)$ with $\init_\aut(v) = \delta(q_\initstate, c(v))$ and $\update_\aut(q,v) = \delta(q,\col(v))$, which is a memory structure for $\arena$. By construction, we have $\update^*(v_0 \cdots v_j) = \delta^*(\col(v_0 \cdots v_j))$. In particular, $\col(v_0 \cdots v_j) \in L(\aut)$ if and only if $\update^*(v_0 \cdots v_j) \in F$.\label{page-autmem}

%% file: content/limit.tex
Recall that $C$ is the finite set of colors used to define winning conditions. The limit of a language~$K \subseteq C^*$ of finite words is
 \[\lim(K) = \set{\alpha_0 \alpha_1 \alpha_2\cdots \in C^\omega \mid \alpha_0 \cdots \alpha_j \in K \text{ for infinitely many }j}\]
 containing all infinite words that have infinitely many prefixes in $K$.
For technical reasons, we require in the following $\epsilon\notin K$.

We call a game of the form~$\game = (\arena, \lim(K))$ a weighted limit game and define the value of a play~$\rho = v_0 v_1 v_2 \cdots $ as
\[
\val_\game(\rho) = \sup_{j \in \nats} \min_{\substack{j' > j \\ \col(v_0\cdots v_{j'}) \in K}} \weight(v_j \cdots v_{j'}),
\]
where $\min \emptyset = \infty$. 
Intuitively, we measure the quality of a winning play by the maximal weight of an infix between two consecutive prefixes whose color sequences are in $K$.  
Note that this value might be $\infty$, even for plays in $\lim(K)$ (see Example~\ref{example-limit-valueunboundedness}) and that it is necessarily $\infty$ if the play is not in $\lim(K)$. Also, let us remark that this definition depends on $K$, not only on $\lim(K)$: It is straightforward to construct languages~$K$ and $K'$ with $\lim(K) = \lim(K')$, but the value functions induced by $K$ and $K'$ differ. Hence, we always make sure that the language~$K$ inducing the value function is clear from context.

\begin{remark}
\label{remark-limit-valuegeneralizeswinning-play}
Let $\game = (\arena, \win)$ be a weighted limit game. Then, $\val_\game(\rho) < \infty$ implies $\col(\rho) \in \win$.	
\end{remark}

Note that the other direction does not hold, as shown in the next example.

\begin{example}
\label{example-limit-valueunboundedness}
For the sake of simplicity, we identify vertices and their color in this example. Hence, 
let $K = \set{v_0,v_1}^* v_1$. Then, $ \lim(K)$ is the set of words having infinitely many occurrences of $v_1$. Now, in a game~$\game$ with winning condition~$\lim(K)$ and a weight function mapping every edge to $1$, $\val_\game(\rho)$ is equal to the supremum over the length of infixes of the form~$v_1v_0^*$ in $\rho$. This may be $\infty$, even if the play~$\rho$ is in $\lim(K)$, e.g., in the play 
\[
\rho= v_0\,
v_1\,
v_0v_0\,
v_1\,
v_0v_0v_0\,
v_1\,
v_0v_0v_0v_0\,
v_1\,
v_0v_0v_0v_0v_0\,
v_1\,
\cdots
 .\]
\end{example}

Given a strategy~$\sigma$ for Player~$0$ and a vertex~$v$, define $\val_\game(\sigma,v) = \sup_{\rho} \val_\game(\rho)$ with the supremum ranging over all plays~$\rho$ that start in $v$ and are consistent with $\sigma$. 
Remark~\ref{remark-limit-valuegeneralizeswinning-play} can be lifted from plays to strategies. 

\begin{remark}
Let $\game = (\arena, \win)$ be a weighted limit game and let $\sigma$ be a strategy for Player~$0$. Then, $\val_\game(\sigma, v) < \infty$ implies that $\sigma$ is a winning strategy for Player~$0$ from $v$ in $\game$.
\end{remark}

Again, the other direction of the implication does not hold, which can be seen by constructing a one-player game where Player~$0$ produces the play from Example~\ref{example-limit-valueunboundedness}.

We say that a strategy~$\sigma$ for Player~$0$ in a weighted limit game~$\game$ is optimal, if it satisfies $\val_\game(\sigma, v) \le \val_\game(\sigma',v)$ for every strategy~$\sigma'$ for Player~$0$ and every vertex~$v$. Note that this definition is a global one, i.e., the strategy has to be better than any other strategy from \emph{every} vertex.

Further, a weighted limit game with winning condition~$\win \subseteq C^\omega$ is regular, if $\win = \lim(L(\aut))$ for some DFA~$\aut$ over $C$. Note that every such language is $\omega$-regular, (in fact it is recognized by $\aut$ when seen as Büchi automaton). In contrast, not every $\omega$-regular language is a regular limit language, e.g., the $\omega$-regular language $(a+b)^*b^\omega $ of words with finitely many $a$ is not a regular limit language. In fact, Landweber showed that the regular limit languages are exactly the languages recognized by deterministic Büchi automata~\cite{DBLP:journals/mst/Landweber69}.

Our main results on regular weighted limit games show that Player~$0$ has an optimal strategy in every such game and how to compute an optimal strategy. 

\begin{theorem}
\label{thm-limit-main}\hfill
\begin{enumerate}
	\item Player~$0$ has an optimal finite-state strategy in every regular weighted limit game.
	\item The problem \myquot{Given an arena~$\arena$ and a DFA~$\aut$, compute an optimal strategy for Player~$0$ in $(\arena, \lim(L(\aut)))$} is solvable in time~$\bigo(\size{V}^3 \cdot \size{E} \cdot \size{Q}^2 \cdot \size{F}^2 ) $, where $(V,E)$ is the graph underlying $\arena$ and $Q$ and $F$ are the sets of states and accepting states of $\aut$ (using the unit-cost model).
\end{enumerate}
\end{theorem}

Before we prove this result, let us comment on one restriction of our model: We only allow nonnegative edge weights. The reason is that it is straightforward to construct a game witnessing that optimal finite-state strategies do not necessarily exist in arenas with negative weights.

\begin{example}
\label{example-negativeweights}
	Consider the game depicted in Figure~\ref{fig-negativeweights}. As Player~$0$ moves at every vertex, we can identify plays and strategies. Also, for the sake of simplicity, we identify vertex names and colors and consider $K = (v_0v_1^*v_2)^*$, i.e., the winning plays are of the form $(v_0v_1^+v_2)^\omega$.
	For every $j > 0$, Player~$0$ has a finite-state strategy to produce the play~$\rho_j = (v_0v_1^jv_2)^\omega$ with $\val_\game(\rho_j) = -j$, which is also the value of the strategy from $v_0$. Hence, she can enforce arbitrarily small values. Furthermore, straightforward pumping arguments show that every finite-state strategy has a bounded value, as it has to leave $v_1$ after a bounded number of steps. 
	
	Altogether, there is no optimal finite-state strategy. 
	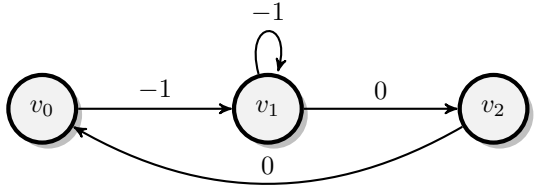
\begin{figure}
	\centering
		\begin{tikzpicture}
			\node[pl0] (s0) at (0,0) {$v_0$};
			\node[pl0] (s1) at (3,0) {$v_1$};
			\node[pl0] (s2) at (6,0) {$v_2$};
			
			\path[->]
			(s0) edge node[above] {$-1$} (s1)
			(s1) edge[loop above] node[above] {$-1$} ()
			(s1) edge node[above] {$0$} (s2)
			(s2) edge[bend left] node[above] {$0$} (s0)
			;
		\end{tikzpicture}
		\caption{The arena for Example~\ref{example-negativeweights}.}
	\label{fig-negativeweights}
	\end{figure}
\end{example}

To prove Theorem~\ref{thm-limit-main}, we first consider the simpler setting of weighted reachability games, i.e., games where a prefix in $K$ has to be reached at least once. This problem is a special case of more general problems that have been considered before (see, e.g.,~\cite{DBLP:conf/concur/BrihayeGHM15,DBLP:journals/mst/KhachiyanBBEGRZ08}). However, these works do not prove all the results we require here. Hence, we discuss in Subsection~\ref{subsection-reachalgo}
a fixed point algorithm computing optimal strategies in reachability games. Then, we use this algorithm as a black box to build another fixed point algorithm computing optimal strategies in weighted limit games (Subsection~\ref{subsection-limitalgo}).

%% file: content/reachalgo.tex
Given a DFA~$\aut$ over $C$ with $\epsilon \notin L(\aut)$, define for a play~$\rho = v_0v_1v_2 \cdots$
\[
\val_\game^R(\rho) = \min_{j \in \nats} \set{\weight(v_0 \cdots v_j) \mid \col(v_0 \cdots v_j) \in L(\aut)},
\]
where $\min \emptyset = \infty$. So, $\val_\game^R(\rho)$ is the weight of the shortest nonempty prefix of $\rho$ whose label sequence is accepted by $\aut$. This also minimizes the accumulated weight, as we only consider nonnegative weights on edges. This definition for plays is lifted to strategies~$\sigma$ for Player~$0$ as for limit games: $\val_\game^R(\sigma, v) = \sup_\rho \val_\game^R(\rho)$ where $\rho$ ranges over all plays starting in the vertex~$v$ that are consistent with $\sigma$. Similarly, optimality of strategies is defined as for limit games.

In the remainder of this section, we show how to compute optimal strategies with respect to $\val_\game^R$, given an arena~$\arena$ and a DFA~$\aut$. 
First, let $\arena \times \mem_{\aut} = (V, V_0, V_1, E, \weight, c)$ be the product of $\arena$ and the memory structure induced by $\aut$ (see Page~\pageref{def-autmem}). Furthermore, let $F$ be the set of vertices of the form~$(v,q)$ where $q$ is an accepting state of $\aut$, i.e., $F$ is a set of vertices of the product arena, \emph{not} the set of accepting states of $\aut$.
	However, reaching a state in $F$ from a vertex of the form $(v, \init(v))$ signifies that the label sequence induced by the play is accepted by $\aut$ (see Page~\pageref{page-autmem}).

A ranking for $\arena \times \mem_{\aut} $ is a mapping $\ranking \colon V \rightarrow \natsclosed$. 
Let $\rankings$ denote the set of all rankings.
We order rankings by defining $r \sqsubseteq r'$ if $r(v) \geq r'(v)$ for all $v \in V$, i.e., $r'$ is \myquot{better} than $r$ if $r'$ assigns ranks that are
pointwise no larger than those of $r$. Hence, the least (and thus the worst) ranking is the one mapping every vertex to $\infty$. Furthermore, there are no infinite strictly ascending chains of rankings, as the ranks only decrease in such a chain, but are always nonnegative.

Next, we define the map~$\lift \colon \rankings \rightarrow \rankings$ via 
\[
\lift(\ranking)(v) = \begin{cases}
	0 &\text{ if $v \in F$,}\\
	\min\set{ \ranking(v), \min_{v' \in vE} \weight(v,v') +\ranking(v') } &\text{ if $v \in V_0 \setminus F$,}\\
	\min\set{ \ranking(v), \max_{v' \in vE} \weight(v,v') +\ranking(v') } &\text{ if $v \in V_1 \setminus F$.}
\end{cases}
\]

We will use $\lift$ to compute the value of an optimal strategy: At vertices in $F$, Player~$0$ has already achieved her goal, i.e., they are assigned a rank of $0$. Now, if it is Player~$0$'s turn at a vertex $v \notin F$, then she has to move to a successor. As she aims to minimize the accumulated weight, she prefers a successor~$v'$ that minimizes the sum of the weight~$\weight(v,v')$ of the edge leading to $v'$ and the rank of $v'$. The reasoning for Player~$1$ is dual: he tries to maximize the accumulated weight. Finally, for technical reasons, we ensure that $\lift$ does never increase a rank via taking the minimum with the old rank of $v$ (which ensures that $\lift$ is monotone).

\begin{remark}
\label{remark-reachability-rankingmonotonicity}
We have $\ranking \sqsubseteq \lift(\ranking)$ for every ranking~$\ranking$.
\end{remark}

Let $\ranking_0$ be the least element of $\rankings$, i.e., the ranking mapping every vertex to $\infty$, and let $\ranking_{j+1} = \lift(\ranking_j)$ for every $j$. Then, we define $r^* = r_n$ for the minimal $n$ with $\ranking_n = \ranking_{n+1}$. Note that such a (least) fixed point~$\ranking^n$ exists due to Remark~\ref{remark-reachability-rankingmonotonicity} and as $\sqsubseteq$ has no infinite strictly ascending chain. From $r^*$ one can derive an optimal strategy for Player~$0$ and the values of such a strategy.

\begin{example}
\label{example-ranking}
Consider the arena depicted in Figure~\ref{figure-rankingexample}, where we mark vertices in $F$ by doubly-lined vertices. We illustrate the computation of the rankings~$\ranking_j$ below the arena, which reaches a fixed point after four applications of $\lift$, i.e., $\ranking_4 = \ranking_5$. Note that the rank of vertex~$v_4$ is updated twice. 

Let us sketch how to extract a strategy for Player~$0$ from the fixed point~$\ranking_4$. Consider, e.g., the vertex~$v_2 \in V_0$. It has rank~$4$ and an edge of weight~$4$ leading to a vertex of rank~$4-4 = 0$, which is the optimal move. In general, every vertex~$v$ of Player~$0$ with finite rank~$\ranking(v)$ has an edge to a successor~$v'$ such that $\ranking(v') = \ranking(v) - \weight(v,v') $. 
Dually, consider the vertex~$v_1 \in V_1$: It has rank~$5$ and every edge leaving $v_1$ goes to a vertex~$v'$ of rank at most $5-\weight(v,v')$. Again, this property is satisfied for every vertex with finite rank.

Hence, using these two properties inductively shows that Player~$0$ has a strategy so that every move from a vertex that is not in $F$ decreases the rank by the weight of the edge taken. Thus, as ranks are nonnegative, a visit to $F$ is guaranteed unless from some point onwards only edges of weight~$0$ are used. However, we will rule this out by ensuring that the target of the edge of weight~$0$ has reached its final rank before the source of the edge, e.g., the successors~$v_2$ and $v_3$ of vertex~$v_3$ have rank~$4$ and the corresponding edge has weight~$0$. However, $v_2$ has reached its final rank one step before $v_3$ has. Ultimately, we show that either the rank or this so-called settling time strictly decreases along every edge taken from a vertex that is not in $F$. As there is no infinite descending chain in this product order, $F$ has to be reached eventually. Using dual arguments, one can define a strategy for Player~$1$ and then show these strategies to be optimal.

In the example, Player~$0$ moves from $v_4$ to $v_3$, from where she moves to $v_2$ and then to $v_0$. This strategy is optimal from every vertex and realizes the value~$\ranking_4(v)$ from every vertex~$v$. For example, the unique play consistent with this strategy starting in $v_4$ has value~$11$. 

It is instructive to compare the computation of the rankings to the attractor computation (see, e.g.,~\cite{GraedelThomasWilke02}): a straightforward induction shows that the $j$-th level of the attractor computation is equal to $\set{v \mid \ranking_{j+1}(v) \neq \infty}$. However, the attractor yields a strategy that minimizes the number of moves necessary to reach $F$ while the rankings minimize the accumulated weight. This difference is witnessed by vertex~$v_4$: the attractor strategy  takes the direct edge to $v_0$ of weight~$99$ while the rankings induce the strategy described above, which realizes a smaller value by taking a longer path through the arena.

\begin{figure}
\centering
\begin{tikzpicture}

\node[pl0,accepting] (v0) {$v_0 $};
\node[pl1,left of = v0] (v1) {$v_1$};
\node[pl0,left of = v1] (v2) {$v_2$};
\node[pl0,left of = v2] (v3) {$v_3$};
\node[pl0,left of = v3] (v4) {$v_4$};
\node[pl1,left of = v4] (v5) {$v_5$};
\node[pl0,left of = v5] (v6) {$v_6$};

\path[->]
(v0) edge[loop right] node[right] {$3$} (v0)
(v1) edge node[above] {$1$} (v0)
(v2) edge[bend left] node[above] {$0$} (v1)
(v1) edge[bend left] node[below] {$1$} (v2)
(v2.north) edge[bend left] node[above] {$4$} (v0.north)
(v3) edge node[above] {$0$} (v2)
(v3) edge[loop above] node {$0$} ()
(v4) edge node[above] {$7$} (v3)
(v5) edge[bend left] node[below] {$2$} (v6)
(v6) edge[bend left] node[above] {$3$} (v5)
(v4) edge[bend right] node[above] {$99$} (v0)
(v5) edge node[above] {$0$} (v4)
(v6) edge[loop left] node[left] {$0$} (v6)
(v3) edge[bend right] node[above] {$5$} (v5)
;
\tikzset{node distance = 2cm};

\node[below of = v0] (0-0) {$\infty$};
\node[below of = v1] (1-0) {$\infty$};
\node[below of = v2] (2-0) {$\infty$};
\node[below of = v3] (3-0) {$\infty$};
\node[below of = v4] (4-0) {$\infty$};
\node[below of = v5] (5-0) {$\infty$};
\node[below of = v6] (6-0) {$\infty$};

\tikzset{node distance = .65cm};

\node[below of = 0-0] (0-1) {$0$};
\node[below of = 0-1] (0-2) {$0$};
\node[below of = 0-2] (0-3) {$0$};
\node[below of = 0-3] (0-4) {$0$};
\node[below of = 0-4] (0-5) {$0$};

\node[below of = 1-0] (1-1) {$\infty$};
\node[below of = 1-1] (1-2) {$\infty$};
\node[below of = 1-2] (1-3) {$5$};
\node[below of = 1-3] (1-4) {$5$};
\node[below of = 1-4] (1-5) {$5$};

\node[below of = 2-0] (2-1) {$\infty$};
\node[below of = 2-1] (2-2) {$4$};
\node[below of = 2-2] (2-3) {$4$};
\node[below of = 2-3] (2-4) {$4$};
\node[below of = 2-4] (2-5) {$4$};

\node[below of = 3-0] (3-1) {$\infty$};
\node[below of = 3-1] (3-2) {$\infty$};
\node[below of = 3-2] (3-3) {$4$};
\node[below of = 3-3] (3-4) {$4$};
\node[below of = 3-4] (3-5) {$4$};

\node[below of = 4-0] (4-1) {$\infty$};
\node[below of = 4-1] (4-2) {$99$};
\node[below of = 4-2] (4-3) {$99$};
\node[below of = 4-3] (4-4) {$11$};
\node[below of = 4-4] (4-5) {$11$};

\node[below of = 5-0] (5-1) {$\infty$};
\node[below of = 5-1] (5-2) {$\infty$};
\node[below of = 5-2] (5-3) {$\infty$};
\node[below of = 5-3] (5-4) {$\infty$};
\node[below of = 5-4] (5-5) {$\infty$};

\node[below of = 6-0] (6-1) {$\infty$};
\node[below of = 6-1] (6-2) {$\infty$};
\node[below of = 6-2] (6-3) {$\infty$};
\node[below of = 6-3] (6-4) {$\infty$};
\node[below of = 6-4] (6-5) {$\infty$};

\tikzset{node distance = 1cm};

\node[left of = 6-0] {$\ranking_0$:};
\node[left of = 6-1] {$\ranking_1$:};
\node[left of = 6-2] {$\ranking_2$:};
\node[left of = 6-3] {$\ranking_3$:};
\node[left of = 6-4] {$\ranking_4$:};
\node[left of = 6-5] {$\ranking_5$:};

\end{tikzpicture}
\caption{The arena for Example~\ref{example-ranking} and the evolution of the corresponding rankings.}
\label{figure-rankingexample}	
\end{figure}
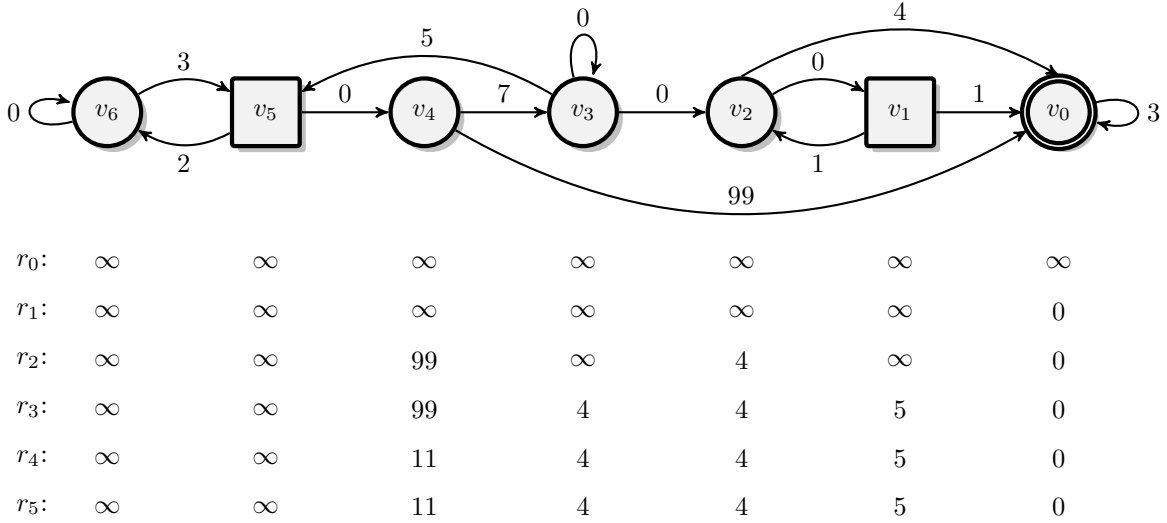

\end{example}

We sketch how to obtain an optimal strategy~$\sigma$ for Player~$0$ from the fixed point~$\ranking^*$, and how $r^*$ and $\sigma$ can be computed in polynomial time. 
To this end, we need to introduce some additional notation.
Consider the sequence~$\ranking_0, \ranking_1, \ldots, \ranking_n = \ranking^*$ as above. Due to Remark~\ref{remark-reachability-rankingmonotonicity}, we have $r_j(v) \ge r_{j+1}(v)$ for every $j$ and every $v$.
The settling time of a vertex~$v$ is defined as $ t_s(v) = \min \set{j \mid \ranking_j(v) = \ranking^*(v)}$, i.e., as the first time $v$ is assigned its final rank~$\ranking^*(v)$.
The construction of an optimal strategy is based on the following results about ranks and settling times, which formalize the intuition given in Example~\ref{example-ranking}.

\begin{lemma}
\label{lemma-reachability-rankingproperties}
Let $v \in V$.
\begin{enumerate}

	\item\label{lemma-reachability-rankingproperties-infty}
	$\ranking^*(v) = \infty$ if and only if $t_s(v) = 0$.
	
	\item\label{lemma-reachability-rankingproperties-F} $v \in F$ implies $\ranking^*(v) = 0$ and $t_s(v) = 1$.
		
	\item\label{lemma-reachability-rankingproperties-playerzero}
	If $v \in V_0 \setminus F$ then $\ranking^*(v) \le \weight(v,v') + \ranking^*(v')$ for all successors~$v' \in vE$. Furthermore, there is some successor~$\overline{v} \in vE$ with $\ranking^*(v) = \weight(v,\overline{v}) + \ranking^*(\overline{v})$. Finally, if $ \ranking^*(v) < \infty $, then $\overline{v}$ can be chosen such that it additionally satisfies $t_s(v) = t_s(\overline{v}) + 1 $. 

	\item\label{lemma-reachability-rankingproperties-playerone}
	If $v \in V_1 \setminus F$ then $\ranking^*(v) \ge \weight(v,v') + \ranking^*(v')$ for all successors~$v' \in vE$.
	 Furthermore, there is some successor~$\overline{v} \in vE$ with $\ranking^*(v) = \weight(v,\overline{v}) + \ranking^*(\overline{v})$. 

	\item\label{lemma-reachability-rankingproperties-playeroneequality}
	If $v \in V_1 \setminus F$ and $\overline{v} \in vE$ with $\ranking^*(v) = \ranking^*(\overline{v}) < \infty$, then $t_s(v) > t_s(\overline{v})$.
	
\end{enumerate}
\end{lemma}
We call successors~$\overline{v}$ as in Items~\ref{lemma-reachability-rankingproperties-playerzero} and \ref{lemma-reachability-rankingproperties-playerone} optimal. 
If Player~$0$ uses an optimal successor, then the rank decreases by the weight of the edge. If this weight is $0$, i.e., the rank stays constant, then the settling time decreases. Similarly, along all edges available to Player~$1$, the  rank decreases at least by the weight of the edge. Again, if that value is $0$, i.e., the rank stays constant, then the settling time decreases.

Using these properties, we define a strategy for Player~$0$ in $\arena$. To this end, we first define a positional strategy~$\sigma'$ for her on $\arena \times \mem_{\aut}$ as follows: at a vertex~$v \in V_0 \setminus F$ move to some optimal successor of $v$. From every vertex~$v \in F \cap V_0$ move to an arbitrary successor. Now, let $\sigma$ be the unique finite-state strategy in $\arena$ implemented by $\mem_{\aut}$ and $\nxt_{\sigma'}$, the next-move function induced by $\sigma'$. 

\begin{lemma}
\label{lemma-reach-sigmaoptimal}
$\sigma$ as defined above is an optimal strategy for Player~$0$ in $\game$.
\end{lemma}

This result is proven in two steps. First, one shows $\val_\game^R(\sigma, v) \le \ranking^*(v, \init(v))$ for every vertex~$v$ of $\arena$, applying the properties posited in Lemma~\ref{lemma-reachability-rankingproperties} inductively. 
Secondly, analogously to the construction of $\sigma$, one constructs a strategy $\tau$ for Player~$1$ satisfying $\val_\game^R(\rho) \ge \ranking^*(v, \init(v))$ for every vertex~$v$ of $\arena$ and every play~$\rho$ starting in $v$ and consistent with $\tau$, which is again proven by applying Lemma~\ref{lemma-reachability-rankingproperties} inductively.

Furthermore, by bounding the settling times of vertices one can show that the fixed point $\ranking^*$ is reached after a linear number of applications of $\lift$.

\begin{lemma}
\label{lemma-reachability-termination}
We have $\ranking^* = \ranking_{\size{\arena} \cdot \size{\aut}+1}$.
\end{lemma}

A simple corollary of the previous lemma yields an upper bound on $\val_\game^R$, which follows from the fact that each application of $\lift$ increases the ranks by no more than the maximal weight of an edge.

\begin{corollary}
\label{coro-reachability-upperboundvalues}
If $\val_\game^R(v) < \infty$ then $\val_\game^R(v) \le \size{\arena} \cdot \size{\aut} \cdot W$, where $W$ is the largest weight in $\arena$. 
\end{corollary}

One can show that the upper bound on the value is tight, e.g., using a game similar to the one presented in Figure~\ref{fig-value-lowerbounds2} on Page~\pageref{fig-value-lowerbounds2}.

%% file: content/limitalgo.tex
Now, we use the fixed point algorithm of the previous subsection to achieve the main goal of this work: solving regular weighted limit games optimally. 
Thus, fix a weighted arena~$\arena$ and a DFA~$\aut$ over $C$ inducing the winning condition~$\lim(K)$ and let $\arena \times \mem_{\aut} = (V, V_0, V_1, E, \weight, c )$ be the product of $\arena$ and the memory structure induced by $\aut$. Furthermore, let $F$ be the set of vertices of the form~$(v,q)$ where $q$ is an accepting state of $\aut$, i.e., $F$ is again a set of vertices of the product arena, \emph{not} the set of accepting states of $\aut$.
	
Recall that $\rankings$ is  the set of rankings~$\ranking \colon V \rightarrow \natsclosed$, which is ordered by $\sqsubseteq$ with $\ranking \sqsubseteq \ranking '$ if and only if $\ranking(v) \ge \ranking'(v)$ for all $v \in V$. Hence, the largest (i.e., best) element of $\rankings$ is the ranking mapping every vertex to $0$. 
We use the operator~$\lift$ defined in Subsection~\ref{subsection-reachalgo} to solve limit games. Recall that $\lift$ allows to compute, for a given set of goal vertices, an optimal strategy that ensures a visit to a goal vertex. However, here we have to treat the set of goal vertices as a parameter because we need to compute optimal strategies for subsets of $F$. Hence, we write $\lift_{F'}$ for $F' \subseteq V$ for the operator
\[
\lift_{F'}(\ranking)(v) = \begin{cases}
	0 &\text{ if $v \in F'$,}\\
	\min\set{ \ranking(v), \min_{v' \in vE} \weight(v,v') +\ranking(v') } &\text{ if $v \in V_0 \setminus F'$,}\\
	\min\set{ \ranking(v), \max_{v' \in vE} \weight(v,v') +\ranking(v') } &\text{ if $v \in V_1 \setminus F'$.}
\end{cases}
\]

All results proven about $\lift$ in Subsection~\ref{subsection-reachalgo} also hold true for $\lift_{F'}$. In particular, we can compute an optimal strategy for Player~$0$ to reach $F'$ and for Player~$1$ to avoid~$F'$ whenever possible, and to maximize the weight, if it is not possible.

The fixed point of $\lift_{F'}$ induces an optimal strategy for Player~$0$ to reach $F'$. However, on vertices in $F'$, from which she reaches $F'$ trivially (i.e., in zero steps), the fixed point does not yield any information on how to reach $F'$ \emph{again}. However, this information can easily be generated from the fixed point. Given an  arbitrary ranking~$\ranking$ and a set $F' \subseteq V$ of vertices, define the completion~$\complete_{F'}(\ranking)$ of $\ranking$ (with respect to $F'$) via
\[
\complete_{F'}(\ranking)(v) = \begin{cases}
\ranking(v) & \text{ if $ v \notin F'$,}\\
\min_{v' \in vE} \weight(v,v') + \ranking(v') & \text{ if $v \in F' \cap V_0$,}\\
\max_{v' \in vE} \weight(v,v') + \ranking(v') & \text{ if $v \in F' \cap V_1$.}
\end{cases}
\]
If $\ranking$ is the least fixed point of $\lift_{F'}$, then $\complete_{F'}(\ranking)$ is obtained from $\ranking$ by assigning to each vertex in $F'$ the minimal weright it takes Player~$0$ to reach $F'$ once more. This is necessary, as we need to reach $F$ infinitely often to win a limit game. The values for all $v \notin F'$ coincide in $\ranking$ and $\complete_{F'}(\ranking)$. 

Recall that the definition of optimal successors in Subsection~\ref{subsection-reachalgo} with respect to the least fixed point $\ranking$ of $\lift_{F'}$ is only defined for vertices in $V \setminus F'$. For $\ranking' = \complete_{F'}(\ranking)$, we can extend this notion to $F'$ as well as follows: a successor~$\overline{v}$ of $v$ in $F'$ is optimal, if $\ranking'(v) = \weight(v, \overline{v}) +\ranking(\overline{v})$.

Now, we again define an operator~$\limitlift$ updating rankings and show that determining a fixed point of the operator induces optimal strategies for both players. Intuitively, the operator tries to reach $F$ with minimal weight, but also has to account for the fact that $F$ has to be reached repeatedly, i.e., the ranks of the vertices reached in $F$ should be as small as possible.
 
Formally, given a ranking~$\ranking$, let $\ranking(F) = \set{ \mathfrak{r}_1 < \mathfrak{r}_2 < \cdots < \mathfrak{r}_k }$, i.e., the $\mathfrak{r}_h$ are the different ranks assigned by $\ranking$ to vertices in $F$. Now, define $F_h = \set{v \in F \mid \ranking (v) \le \mathfrak{r}_h}$ for $1 \le h \le k$, i.e., we order the vertices in $F$ into a hierarchy~$F_1 \subseteq  F_2 \subseteq  \cdots \subseteq F_k$ according to their rank with the intuition that smaller ranks are preferable for Player~$0$. Let $\ranking_h'$ be the least fixed point of $\lift_{F_h}$ for $1 \le h \le k$ and $\ranking_h'' = \complete(\ranking_h')$. Then, we define the ranking~$\limitlift(\ranking)$ via
\[
\limitlift(\ranking)(v) = \min_{1 \le h \le k}(\max\set{\ranking(v), \ranking_{h}''(v), \mathfrak{r}_h }) ,
\] 
i.e., to compute the new rank of $v$ we take into account the old rank and then minimize over the maximum of the weight to reach some $F_h$ and the maximal old rank of the vertices in $F_h$, which indicates (in the fixed point) how costly it is to reach $F$ repeatedly from this vertex.

\begin{remark}
\label{remark-limit-rankingmonotonicity}
We have $\ranking \sqsupseteq\limitlift(\ranking) $ for every ranking~$\ranking$. 
\end{remark}

Now, let $\ranking_0$ be the ranking mapping every vertex to $0$, i.e., the $\sqsubseteq$-largest ranking, and define $\ranking_{j+1} = \limitlift(\ranking_j)$ for every $j > 0$.

\begin{example}
 \label{example-ranking-limit}
Consider the game in Figure~\ref{figure-rankingexample-limit} and focus on vertex~$v_1$. Its rank is updated from its initial value of $0$ to $2$ (because the vertex~$v_2$ in $F$ can be reached with weight~$2$) and then $3$ (because reaching $F$ once more from $v_2$ incurs weight~$3 = \max\set{2,3}$) and then to $7$ (as $F$ is no longer reachable from $v_3$, but from $v_0$ which incurs weight~$\max\set{4,7}$).

 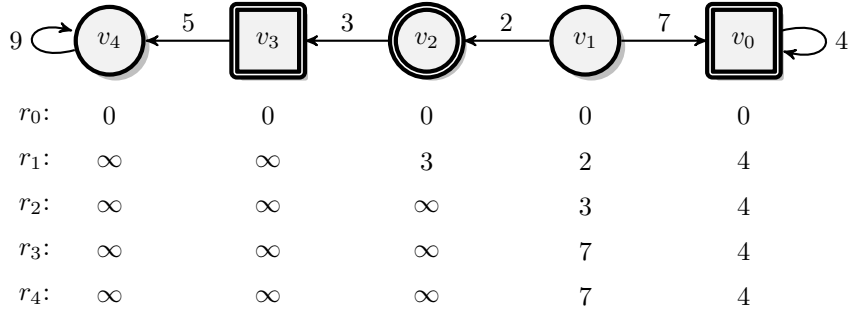
\begin{figure}
\centering
\begin{tikzpicture}

\node[pl1,accepting] (v0) {$v_0 $};
\node[pl0,left of = v0] (v1) {$v_1$};
\node[pl0,accepting,left of = v1] (v2) {$v_2$};
\node[pl1,accepting,left of = v2] (v3) {$v_3$};
\node[pl0,left of = v3] (v4) {$v_4$};

\path[->]
(v0) edge[loop right] node[right] {$4$} (v0)
(v1) edge node[above] {$7$} (v0)
(v1) edge node[above] {$2$} (v2)
(v2) edge node[above] {$3$} (v3)
(v3) edge node[above] {$5$} (v4)
(v4) edge[loop left] node[left] {$9$} ()
;

\tikzset{node distance = 1cm};

\node[below of = v0] (0-0) {$0$};
\node[below of = v1] (1-0) {$0$};
\node[below of = v2] (2-0) {$0$};
\node[below of = v3] (3-0) {$0$};
\node[below of = v4] (4-0) {$0$};

\tikzset{node distance = .6cm};

\node[below of = 0-0] (0-1) {$4$};
\node[below of = 0-1] (0-2) {$4$};
\node[below of = 0-2] (0-3) {$4$};
\node[below of = 0-3] (0-4) {$4$};

\node[below of = 1-0] (1-1) {$2$};
\node[below of = 1-1] (1-2) {$3$};
\node[below of = 1-2] (1-3) {$7$};
\node[below of = 1-3] (1-4) {$7$};

\node[below of = 2-0] (2-1) {$3$};
\node[below of = 2-1] (2-2) {$\infty$};
\node[below of = 2-2] (2-3) {$\infty$};
\node[below of = 2-3] (2-4) {$\infty$};

\node[below of = 3-0] (3-1) {$\infty$};
\node[below of = 3-1] (3-2) {$\infty$};
\node[below of = 3-2] (3-3) {$\infty$};
\node[below of = 3-3] (3-4) {$\infty$};

\node[below of = 4-0] (4-1) {$\infty$};
\node[below of = 4-1] (4-2) {$\infty$};
\node[below of = 4-2] (4-3) {$\infty$};
\node[below of = 4-3] (4-4) {$\infty$};

\tikzset{node distance = 1cm};

\node[left of = 4-0] {$\ranking_0$:};
\node[left of = 4-1] {$\ranking_1$:};
\node[left of = 4-2] {$\ranking_2$:};
\node[left of = 4-3] {$\ranking_3$:};
\node[left of = 4-4] {$\ranking_4$:};

\end{tikzpicture}
\caption{The arena for Example~\ref{example-ranking-limit} and the evolution of the corresponding rankings.}
\label{figure-rankingexample-limit}	
\end{figure}

\end{example}

To begin our proof of correctness, we show that the ranks assigned by the $\ranking_j$ are bounded by some polynomial that only depends on $\arena$ and $\aut$ (but is exponential if weights are encoded in binary). In particular, this implies that there is some  $n$ such that $\ranking_n = \ranking_{n+1}$. Again, we denote $\ranking_n$ for the smallest such $n$ as $\ranking^*$ (which is the greatest fixed point of $\limitlift$). 

\begin{lemma}
\label{lemma-limit-termination-naive}
Let $v \in V$ and $j \ge 0$. 
If $\ranking_j(v)< \infty$ then $\ranking_j(v) \le (\size{\arena} \cdot \size{\aut} +1) \cdot W$, where $W$ is the largest weight in $\arena$. 
\end{lemma}

In the following, consider the  application of $\limitlift$ to $\ranking^*$: let the $\mathfrak{r}_h$, $F_h$, $\ranking_h'$, and $\ranking_h''$ be computed with respect to $\ranking^*$ as described above. For every $v \in V$, let $h(v)$ be such that
\[\ranking^*(v) = \min_{1 \le h \le k}(\max\set{\ranking^*(v), \ranking_{h}''(v), \mathfrak{r}_h }) = \max\set{\ranking^*(v), \ranking_{h(v)}''(v), \mathfrak{r}_{h(v)} }.\] If there are several possible values for $h(v)$, we pick the smallest one with this property (although this is inconsequential). 

Next, we define a finite-state strategy~$\sigma'$ for Player~$0$ in $\arena \times \mem_\aut$ implemented by a memory structure~$\mem' = (M', \init', \update')$ with $M' = \set{1, \cdots, k}$, $\init'(v) = h(v) $, and $\update'(h,v) = h$, if $v \notin F_h$, and $\update'(h,v) = \init'(v)$, if $v \in F_h$. Thus, the memory is initalized to $h(v)$ when starting at $v$ and stays constant until a vertex $v' \in F_{h(v)}$ is visited. While moving to $v'$, the memory is again initialized to $h(v')$ and stays constant until $F_{h(v')}$ is visited. This procedure is repeated ad infinitum.  It remains to define the next-move function: $\nxt'(v,h)$ is an optimal successor of $v$ with respect to $\ranking_h'$, if $v \notin F_h$, and an optimal successor of $v$ with respect to $\ranking_h''$, if $v \in F_h$. 
 Let $\sigma'$ be the strategy implemented by $\mem'$ and $\nxt'$ in $\arena \times \mem$ and let $\sigma$ be the strategy induced by $\mem$ and $\sigma'$ in $\arena$.

\begin{lemma}
\label{lemma-limit-valsigma}
We have $\val_{\game}(\sigma, v) \le \ranking^*(v, \init(v)) $ for every $v$ in $\arena$.	
\end{lemma}

Recall that we have a sequence~$\ranking_0 \sqsupseteq \ranking_1 \sqsupseteq \cdots  \sqsupseteq \ranking_n = \ranking_{n+1} = \ranking^*$ of rankings with $\ranking_{j+1} = \limitlift(\ranking_j)$ for every $j \leq n$. 
Here, we define the settling time~$t_s(v)$ of a vertex~$v \in V$ as the minimal~$j$ with $\ranking_j(v) = \ranking^*(v)$. 

\begin{remark}
\label{remark-limit-settlingtimes}
$\ranking^*(v) > 0$ implies $t_s(v) > 0$ and $\ranking_{t_s(v)-1}(v) < \ranking_{t_s(v)}(v)$.
\end{remark}

Next, we define a finite-state strategy~$\tau'$ for Player~$1$ in $\arena \times \mem_\aut$ implemented by a memory structure~$\mem' = (M', \init'\, \update')$ with $M' = V$, $\init'(v) = v$, and $\update'(v,v') = v $, if $v' \notin F$, and $\update'(v,v') = \init'(v') = v' $, if $v' \in F$ (recall that the first argument of an update function is the current memory state and the second one a vertex).
To define the next-move function, we distinguish three types of vertices~$v \in V$.

 We say $v$ is of type zero, if $\ranking^*(v) = 0$. If this is not the case, i.e., if $\ranking^*(v) > 0$, then we have
 \begin{equation}
 \label{eq-pl1limit}	
 \ranking^*(v) = \ranking_{t_s(v)}(v) = \min_h (\max\set{ \ranking_h''(v), \mathfrak{r}_h})
 \end{equation}
 due to Remark~\ref{remark-limit-settlingtimes},
where the $\ranking_h''$ and $\mathfrak{r}_h$ are computed with respect to $\ranking_{t_s(v)-1}$. 
Now, we say $v$ is of type one, if there is an $h$ such that $\ranking^*(v) = \ranking_h''(v)$. Then, we define $h(v)$ to be the maximal $h$ with this property.

Finally, if there is no $h$ with $\ranking^*(v) = \ranking_h''(v)$, then we must have $\ranking^*(v) = \mathfrak{r}_h$ for some $h$. Due to the $\mathfrak{r}_h$ being strictly increasing, there is a unique $h = h(v)$ with this property. In this case, we say $v$ is of type two. 

Now, if $v$ is of type zero, then we define $\nxt'(v', v)$ to be an arbitrary successor of $v'$ (recall that the first argument of a next-move function is the current vertex and the second one the current memory state). If $v$ is of type one, then we define $\nxt'(v' v)$ to be an optimal successor of $v'$ with respect to $\ranking_{h(v)}''$. Finally, if $v$ is of type two and we have $h(v) = 1$, then let $\nxt'(v' v)$ be an arbitrary successor of $v'$. On the other hand, if $v$ is of type two and we have $h(v) > 1$, then let $\nxt'(v' v)$ be an optimal successor of $v'$ with respect to $\ranking_{h(v)-1}''$.
 Let $\tau'$ be the strategy implemented by $\mem'$ and $\nxt'$ in $\arena \times \mem$ and let $\tau$ be the strategy induced by $\mem$ and $\tau'$ in $\arena$. 

\begin{lemma}
\label{lemma-limit-valtau}
We have $\val_{\game}(\tau, v) \ge \ranking^*(v, \init(v)) $ for every $v$ in $\arena$.	
\end{lemma}

Lemmata~\ref{lemma-limit-valsigma} and \ref{lemma-limit-valtau} imply that $\sigma$ and $\tau$ are optimal strategies (where optimality of Player~$1$ strategies is defined as expected), i.e., the first part of our main theorem is proven.

The construction of $\tau$ also yields an upper bound on the number of iterations of $\limitlift$ that are necessary to reach the fixed point.

\begin{lemma}
\label{lemma-limit-termination}
We have $\ranking^* = \ranking_{\size{F}+1}$.
\end{lemma}

It remains to determine the overall running time of our algorithm.
Recall that we have defined $F$ to be the product of the set of  vertices of the arena~$\arena$ and the accepting states of $\aut$.
Untangling the construction above shows that the fixed point of $\limitlift$ can be computed in time~$\bigo(n^3 e s^2 f^2 ) $, where $n$ and $e$ are the number of vertices and edges of $\arena$ and  $s$ and $f$ are the number of  states and accepting states of $\aut$: Due to Lemma~\ref{lemma-limit-termination}, it takes at most $\size{F}+1 = n \cdot f+1$ applications of $\limitlift$ to reach the fixed point, each taking at most $\size{F}$ computations of a fixed point of $\lift_{F'}$. Each of these takes at most $n\cdot s +1$ applications of $\lift$, which each takes time $e\cdot s$ in the unit-cost model. 
 
Note that optimal strategies for Player~$0$ in $\arena$ are implemented by memory structures that do not need to keep track of weights of play prefixes, only pairs of vertices and states. The following corollary gives an upper bound on the size and quality of optimal strategies. 

\begin{lemma}
\label{lemma-limit-upperbounds}
Let $\game = (\arena, \lim(L(\aut)))$ be a weighted reachability game with $n$ vertices and largest weight~$W$, and let $s$ and $f$ be the number of states and accepting states of $\aut$. Then, Player~$0$ has an optimal strategy for $\game$ of size~$nsf$ with $\val_{\game}(v) \le (ns +1)\cdot W$ for all vertices~$v$ with $\val_{\game}(v) < \infty$. 
\end{lemma}

The following example shows that both the upper bound on the memory size and the upper bound on the value of an optimal strategy are (almost) tight.

\begin{example}
\label{example-limit-lowerbounds} 
\mbox{}\hfill
\begin{enumerate}
	\item \label{example-limit-lowerbounds-memory}
We begin with the lower bound on the memory.
Consider the arena~$\arena_n$ and the automaton~$\aut_s$ (for $n>0$ and $s>1$) depicted in Figure~\ref{fig-limit-lowerbounds} inducing the game~$\game_{n,s}$. The automaton accepts the language~$a(a^{s-1}b)^*c$. Note that we can identify (winning) strategies for Player~$0$ with (winning) plays, as all vertices are controlled by Player~$0$. Also, from every vertex~$v_j$ there is a unique play (strategy)~$\rho_j = v_jv^{s-1}v' (v_j')^\omega$ with $\val_{\game_{n,s}}(\rho_j) = n+1+j$. Every other play starting in $v_j$ has a larger value. Hence, there is a unique optimal strategy for Player~$0$, which, for every $j$, yields the play~$\rho_j$ when starting in $v_j$.

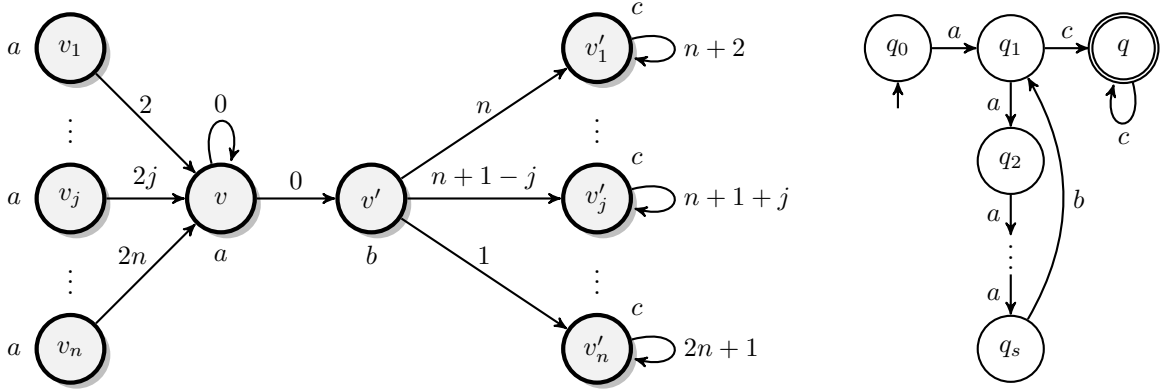
\begin{figure}
\centering
\begin{tikzpicture}[node distance = .75cm]

\node[pl0] (s1) at (-4,2) {$v_1$};
\node (sdotsa) at (-4,1) {$\vdots$};
\node[pl0] (sj) at (-4,0) {$v_j$};
\node (sdotsb) at (-4,-1) {$\vdots$};
\node[pl0] (sm) at (-4,-2) {$v_n$};

\node[pl0] (v1) at (-2,0) {$v$};
\node[pl0] (v1p) at (0,0) {$v'$};

\node[pl0] (t1) at (3,2) {$v_1'$};
\node (tdotsa) at (3,1) {$\vdots$};
\node[pl0] (tj) at (3,0) {$v_j'$};
\node (tdotsb) at (3,-1) {$\vdots$};
\node[pl0] (tm) at (3,-2) {$v_n'$};

\node[below of = v1] {$a$};
\node[below of = v1p] {$b$};

\node[left of = s1] {$a$};
\node[left of = sj] {$a$};
\node[left of = sm] {$a$};

\node[above right of = t1] {$c$};
\node[above right of = tj] {$c$};
\node[above right of = tm] {$c$};

\path[->]
(s1) edge node[above] {$2$} (v1)
(sj) edge node[above] {$2j$} (v1)
(sm) edge node[xshift=-5pt,above] {$2n$} (v1)

(v1) edge[loop above] node[above] {$0$} ()
(v1) edge node[above] {$0$} (v1p)

(v1p) edge node[above] {$n$} (t1)
(v1p) edge node[above] {$n+1-j$} (tj)
(v1p) edge node[above] {$1$} (tm)

(t1) edge[loop right] node[right] {$n+2$} ()
(tj) edge[loop right] node[right] {$n+1+j$} ()
(tm) edge[loop right] node[right] {$2n+1$} ()
;

\node[state] (q0) at (7,2) {$q_0$};
\node[state] (q1) at (8.5,2) {$q_1$};
\node[state] (q2) at (8.5,0.5) {$q_2$};

\node[inner sep = 7pt,label={[yshift=2pt]center:$\vdots$}] (qdots) at (8.5,-.75) {};

\node[state] (qn) at (8.5,-2) {$q_{s}$};
\node[state,accepting] (q) at (10,2) {$q$};

\path[->]
(q1) edge node[left] {$a$} (q2)
(q2) edge node[left] {$a$} (qdots)
(qdots) edge node[left] {$a$} (qn)
(qn) edge[bend right] node[right] {$b$} (q1)
(q1) edge node[above] {$c$} (q) 
(q0) edge node[above] {$a$} (q1)
(7,1.2) edge (q0)
(q) edge[loop below] node[below] {$c$} ()
;

\end{tikzpicture}	
\caption{The arena~$\arena_n$ (left) and the automaton~$\aut_s$ (right) for the lower bounds in Example~\ref{example-limit-lowerbounds}.\ref{example-limit-lowerbounds-memory}. Here, $a$, $b$, and $c$ are the colors of the vertices. Furthermore, all missing transitions of the automaton lead to a rejecting sink state that is not drawn for the sake of readability.}
\label{fig-limit-lowerbounds}
\end{figure}

Furthermore, standard pumping arguments show that every strategy for Player~$0$ yielding, for every $j$, the play~$\rho_j$ when starting at $v_j$ has at least $n(s-1)$ states, which are required to reach $v_j'$ when starting at $v_j$ and to be able to traverse the self-loop at the vertex~$v$ exactly $n-2$ times, as required by the winning condition. Note that this lower bound does not take the number of accepting states into account, i.e., it is not completely tight.

\item \label{example-limit-lowerbounds-value} Next, we consider the lower bound on the value of an optimal strategy for Player~$0$. 
Figure~\ref{fig-value-lowerbounds2} depicts an arena~$\arena_m$ and a DFA~$\aut_n$ (for $m>1$ and $n>1$), which accepts the language~$((a^{n-1}b)^*c)^*$. 
Note that we can identify (winning) strategies for Player~$0$ with (winning) plays, as all vertices are controlled by Player~$0$.
Actually, there is a unique winning play (i.e., winning strategy) for Player~$0$ starting in $v_1$, i.e., the play 
\[((v_1)^{s-1}v_1'  (v_2)^{s-1}v_2' \cdots (v_n)^{s-1}v_n' v)^\omega\] with value~$mnW$. Hence, the value of an optimal strategy from $v_1$ is $mnW$.

\begin{figure}
\centering
\begin{tikzpicture}[node distance = .75cm]
\node[pl0] (v1) at (0,0) {$v_1$};
\node[pl0] (v1p) at (2,0) {$v_1'$};

\node[pl0] (v2) at (4,0) {$v_2$};
\node[pl0] (v2p) at (6,0) {$v_2'$};

\node (dots) at (8,0) {$\cdots$};

\node[pl0] (vm) at (10,0) {$v_n$};
\node[pl0] (vmp) at (12,0) {$v_n'$};
\node[pl0] (v) at (14,0) {$v$};

\node[below of = v1] {$a$};
\node[below of = v1p] {$b$};
\node[below of = v2] {$a$};
\node[below of = v2p] {$b$};
\node[below of = vm] {$a$};
\node[below of = vmp] {$b$};
\node[below of = v] {$c$};

\path[->]
(v1) edge[loop above] node[above] {$W$} ()
(v1) edge node[above] {$W$} (v1p)
(v1p) edge node[above] {$W$} (v2)
(v2) edge[loop above] node[above] {$W$} ()
(v2) edge node[above] {$W$} (v2p)
(v2p) edge node[above] {$W$} (dots)

(dots) edge node[above] {$W$} (vm)
(vm) edge[loop above] node[above] {$W$} ()
(vm) edge node[above] {$W$} (vmp)
(vmp) edge node[above] {$W$} (v)
(v) edge[bend right] node[below] {$W$} (v1);

\node[state,accepting] (q0) at (2,-2) {$q_0$};

\node[state] (q1) at (4,-2) {$q_1$};
\node[state] (q2) at (6,-2) {$q_2$};

\node[] (qdots) at (8,-2) {$\cdots$};

\node[state] (qn) at (10,-2) {$q_{s}$};

\path[->]
(q1) edge node[above] {$a$} (q2)
(q2) edge node[above] {$a$} (qdots)
(qdots) edge node[above] {$a$} (qn)
(qn) edge[bend right] node[below] {$b$} (q1)
(q1) edge[] node[above] {$c$} (q0) 
(q0) edge[bend right] node[below] {$a$} (q2)
(4,-1.1) edge (q1)
;

\end{tikzpicture}	
\caption{The arena~$\arena_n$ (top) and the automaton~$\aut_s$ (bottom) for the lower bounds in Example~\ref{example-limit-lowerbounds}.\ref{example-limit-lowerbounds-value}. Here, $W$ is an arbitrary nonnegative integer and $a$, $b$, and $c$ are the colors of the vertices. Furthermore, all missing transitions of the automaton lead to a rejecting sink state that is not drawn for the sake of readability.}
\label{fig-value-lowerbounds2}
\end{figure}
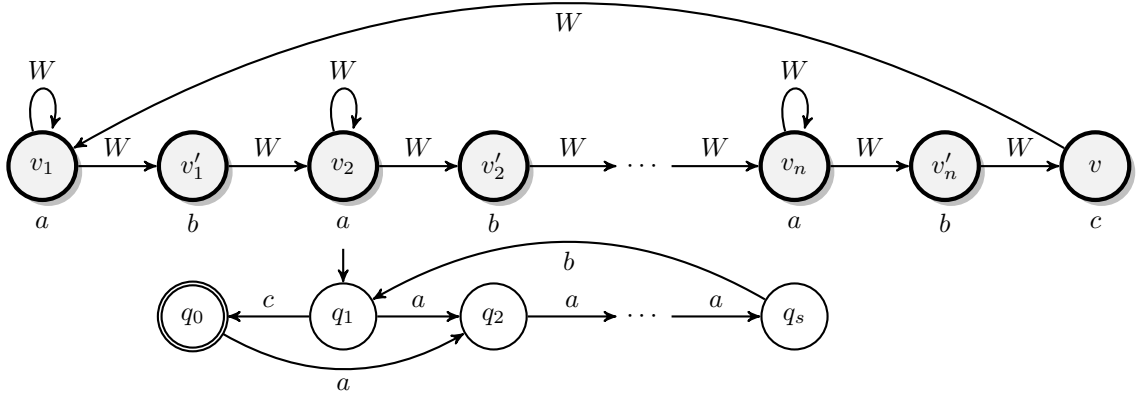
\end{enumerate}
\end{example}

The lower bound on the value presented above is tight while the lower bound on the memory is off by a factor of $f$, where $f$ is the number of accepting states of the automaton. We expect that the upper bound can be improved by removing the factor~$f$ by exploiting some monotonicity properties. In particular, this should be true in the case where we are not constructing a uniform optimal strategy, i.e., one that is optimal from every vertex. Recall the game presented in Example~\ref{example-limit-lowerbounds}.\ref{example-limit-lowerbounds-memory}: here, the factor~$n$ in the memory requirement is due to the fact that the strategy intuitively has to memorize the vertex~$v_j$ the play starts in in order to move to the corresponding $v_j'$ to achieve the optimal value. On the other hand,  a strategy that is only optimal from some fixed~$v_j$ does not have to store the initial vertex but can instead always move to $v_j'$ and thus only needs $v-1$ memory states. Whether the upper bound can be improved in this setting is left open for further work.

%% file: content/furtherwork.tex
The (qualitative) winning region~$W_i(\game)$ of Player~$i$ in a regular weighted limit game~$\game$ contains all vertices~$v$ from which Player~$i$ has a winning strategy. 
In the previous section, we have considered a quantitative notion of winning by measuring the quality of strategies. For finite arenas, it turns out that our quantitative notion is a refinement of the qualitative one.

\begin{lemma}
\label{lemma-refinement}
Let $\game = (\arena, \lim(L(\aut)))$ be a regular weighted limit game and let $\sigma$ be an optimal strategy for Player~$0$ in $\game$. Then, $W_0(\game) = \set{v \mid \val_\game(\sigma,v) < \infty}$ and $W_1(\game) = \set{v \mid \val_\game(\sigma,v) = \infty}$.
\end{lemma}

The previous refinement result relies on the finiteness of the arena. In fact, it is no longer valid in infinite arenas, even in very simple ones with unit weights. 

\begin{example}
\label{example-infinite}
Consider the infinite arena presented in Figure~\ref{fig-infinite} and $K = (ab^+c^+)^*ab^*$, i.e., Player~$0$ wins every play starting in the vertex colored by $a$. Furthermore, the value of a play is equal to the length of the longest infix with label sequence in $c^*a$. 

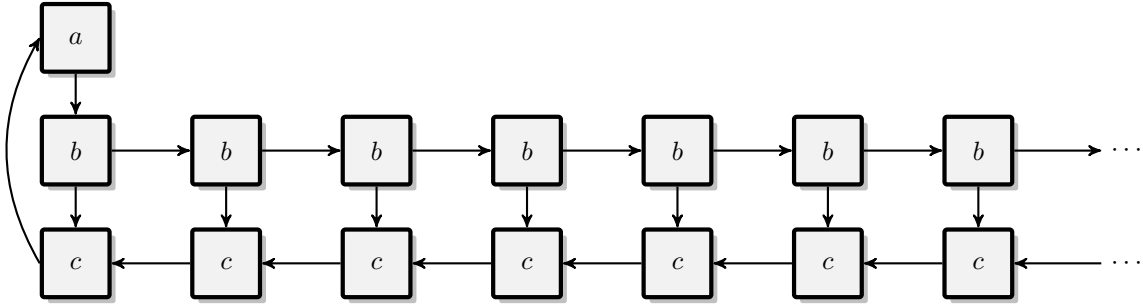
\begin{figure}
	\centering
	\begin{tikzpicture}
		\node[pl1] (v) at (0,0) {$a$};
		\node[pl1] (u0) at (0,-1.5) {$b$};
		\node[pl1] (u1) at (2,-1.5) {$b$};
		\node[pl1] (u2) at (4,-1.5) {$b$};
		\node[pl1] (u3) at (6,-1.5) {$b$};
		\node[pl1] (u4) at (8,-1.5) {$b$};
		\node[pl1] (u5) at (10,-1.5) {$b$};
		\node[pl1] (u6) at (12,-1.5) {$b$};
		\node[] (u7) at (14,-1.5) {$\cdots$};

		\node[pl1] (d0) at (0,-3) {$c$};
		\node[pl1] (d1) at (2,-3) {$c$};
		\node[pl1] (d2) at (4,-3) {$c$};
		\node[pl1] (d3) at (6,-3) {$c$};
		\node[pl1] (d4) at (8,-3) {$c$};
		\node[pl1] (d5) at (10,-3) {$c$};
		\node[pl1] (d6) at (12,-3) {$c$};
		\node[] (d7) at (14,-3) {$\cdots$};
				
		\path[->]
		(v) edge (u0)
		(u0) edge (u1)
		(u1) edge (u2)
		(u2) edge (u3)
		(u3) edge (u4)
		(u4) edge (u5)
		(u5) edge (u6)
		(u6) edge (u7)

		(u0) edge (d0)
		(u1) edge (d1)
		(u2) edge (d2)
		(u3) edge (d3)
		(u4) edge (d4)
		(u5) edge (d5)
		(u6) edge (d6)

		(d7) edge (d6)		
		(d6) edge (d5)		
		(d5) edge (d4)		
		(d4) edge (d3)		
		(d3) edge (d2)		
		(d2) edge (d1)		
		(d1) edge (d0)		
		(d0.west) edge[bend left] (v.west)		
		;
		
	\end{tikzpicture}
\caption{The arena for Example~\ref{example-infinite}. Vertices are labeled by their colors and every edge has weight~$1$.}
\label{fig-infinite}
\end{figure}

Now consider the play~$\rho$ with coloring
\[ abc\, abbcc \, abbbccc \, abbbbcccc \, abbbbbccccc \cdots .  \]
It is winning for Player~$0$, has value~$\infty$ (as the length of $c$-blocks is unbounded), and consistent with every strategy for Player~$0$, as Player~$1$ moves at every vertex. 

Hence, although Player~$0$ wins from the vertex with label~$a$, she does not have a strategy with finite value from this vertex. 
\end{example}

Note that the graph underlying the arena in Example~\ref{example-infinite} is a configuration graph of a one-counter machine, a particularly simple class of infinite graphs with many desirable decidability properties (see, e.g.,~\cite{Serre06} for games on such graphs). 
Nevertheless, quantitative winning no longer refines qualitative winning. 

As mentioned above, the proof of the refinement lemma relies crucially on the finiteness of the arena, which yields the upper bound on the values of an optimal strategy. Hence, on infinite arenas, there are three classes of vertices: those from which Player~$0$ can win with a bounded value, those from which she can win, but not with a bounded value, and those from which she cannot win at all. Thus, the landscape for infinite arenas is, in a sense, much more interesting than for finite arenas and being able to win even with a finite value is more useful than just being able to win.

%% file: content/relatedwork.tex
Quantitative infinite-duration games have received considerable attention, e.g., in the form of games with mean-payoff conditions~\cite{DBLP:journals/dam/BjorklundV07,EhrenfeuchtMycielski79,Puri/95/simprove, DBLP:journals/tcs/ZwickP96} and other payoff conditions~\cite{DBLP:conf/concur/BrihayeGHM15,DBLP:conf/mfcs/GimbertZ04,DBLP:journals/tcs/ZwickP96}, energy conditions~\cite{DBLP:journals/acta/BouyerMRLL18,DBLP:journals/acta/ChatterjeeRR14,DBLP:conf/birthday/JuhlLR13,TV87}, quantitative logics for specifying winning conditions~\cite{DBLP:journals/tocl/AlurETP01,DBLP:journals/iandc/FaymonvilleZ17,DBLP:journals/fmsd/KupfermanPV09,DBLP:journals/tcs/Zimmermann13,DBLP:journals/acta/Zimmermann18}, variations of the classical parity condition~\cite{DBLP:journals/tcs/ChatterjeeD12,DBLP:journals/tocl/ChatterjeeHH09,DBLP:conf/lics/ChatterjeeHJ05,DBLP:journals/corr/abs-1207-0663,DBLP:journals/lmcs/ScheweWZ19}, and other models~\cite{DBLP:conf/cav/BloemCHJ09,DBLP:conf/cav/BrazdilCKN12,DBLP:journals/iandc/BruyereFRR17}. Weighted limit games are related to some of these models.

In particular, the problem of determining the value of an optimal strategy in a weighted limit game is related to the optimal cover problem for one-dimensional consumption games~\cite{DBLP:conf/cav/BrazdilCKN12}. Such a game is also played in a weighted arena and while an edge with weight~$w$ is traversed, a battery is discharged by $w$ units. Furthermore, there are special edges that allow to recharge the battery to an arbitrary amount. Now, the optimal cover problem asks to compute the minimal battery capacity that allows Player~$0$ to play indefinitely without ever completely depleting the battery. 

As long as the arena does not contain any cycles consisting only of edges with weight $0$, one can turn a weighted limit game into a consumption game: After every visit to a vertex in $F$, the battery is recharged and then drained by the weight along the edges until $F$ is visited again. Now, one can show that the minimal sufficient capacity for the battery corresponds to the value of an optimal strategy. However, in the presence of cycles of weight~$0$, this correspondence no longer holds, as such a cycle is sufficient for Player~$0$ to not drain the battery, while this is not sufficient in a weighted limit game if the cycle does not contain a vertex from $F$. Formulated differently: consumption games have a safety winning condition while a limit game has a liveness condition.\footnote{Note that a visit to $F$ could be enforced by having a second dimension that implements a countdown timer that is decremented along each edge.}

On the other hand, synthesis of optimal strategies in weighted limit games can be seen as a special case of the optimization problem for Prompt-LTL with costs~\cite{DBLP:journals/acta/Zimmermann18}.\footnote{Prompt-LTL with costs is the fragment of Parametric LTL with costs allowing only one parameter, which is introduced in~\cite{DBLP:journals/acta/Zimmermann18} without a name.} This is an extension of classical LTL~\cite{DBLP:conf/focs/Pnueli77} by the prompt-eventually~$\mathbf{F}_P$~\cite{DBLP:journals/fmsd/KupfermanPV09}: The formula~$\mathbf{F}_P \varphi$ holds with respect to a bound~$k$ on some weighted trace~$\pi$, if $\pi$ can be decomposed into $\pi = \pi_0 \pi_1$ such that the weight of $\pi_0$ is at most $k$ and $\pi_1$ satisfies $\varphi$ with respect to $k$. Intuitively, $\varphi$ has to be satisfied within a prefix of weight at most $k$.
Now, the formula~$\mathbf{G}\mathbf{F}_P a$ with respect to a bound~$k$ expresses that the atomic proposition~$a$ holds infinitely often and that the weight between consecutive occurrences is bounded by $k$. So, computing the minimal $k$ for which Player~$0$ has a winning strategy for the game with winning condition~$\mathbf{G}\mathbf{F}_P a$, where $a$ holds exactly at the vertices in $F$, yields the value of an optimal strategy. Furthermore, a witnessing winning strategy can be computed~\cite{DBLP:journals/acta/Zimmermann18}. 

Finally, weighted limit games can be seen as a special case of two-color parity games with costs~\cite{DBLP:journals/corr/abs-1207-0663} (with binary encoding~\cite{DBLP:journals/lmcs/Weinert017}), a variant of parity games where Player~$0$ aims to minimize the weight between the occurrences of odd colors and the next larger even color. An optimal strategy for the parity game with costs~\cite{DBLP:journals/lmcs/Weinert017} is also optimal for the weighted limit game.

 However, all three approaches do not yield the fine-grained complexity analysis presented here, e.g., tight upper and lower bounds on the memory requirements and values of optimal strategies.

%% file: content/conc.tex
In this work, we have considered the problem of computing optimal strategies in regular weighted limit games. Such strategies always exist in finite arenas, and are efficiently computable by a fixed point algorithm. Furthermore, we have shown that allowing negative weights leads to games without optimal strategies and how the relation between qualitative and quantitative winning is affected by considering infinite arenas.

The case of infinite arenas is also a promising direction for further work. We conjecture that our fixed point characterization can be lifted to limit games in infinite arenas as well, with some minor adaptions to account for infinite branching and using transfinite induction to obtain the fixed points. However, these are no longer effective, due to the infiniteness of the arena. Instead, it seems promising to consider saturation-based methods~\cite{Cachat02,CarayolH18}. 

Another direction for further work is concerned with more general definitions for the value of a play. Here, we have accumulated the weight of certain infixes. Instead, one could, e.g.,  consider the average weight of these infixes. 

Finally, another promising direction for further work concerns quantitative winning conditions, e.g., limit conditions, in games with imperfect information~\cite{DoyenRaskin11}.

%% file: content/appendix.tex
In this appendix, we present the proof omitted in the main part.

\subsection{Proofs Omitted in Subsection~\ref{subsection-reachalgo}}

\subsubsection{Proof of Lemma~\ref{lemma-reachability-rankingproperties}}

Recall that we need to prove the following statements.

\begin{enumerate}

	\item\label{lemma-reachability-rankingproperties-infty}
	$\ranking^*(v) = \infty$ if and only if $t_s(v) = 0$.
	
	\item\label{lemma-reachability-rankingproperties-F} $v \in F$ implies $\ranking^*(v) = 0$ and $t_s(v) = 1$.
		
	\item\label{lemma-reachability-rankingproperties-playerzero}
	If $v \in V_0 \setminus F$ then $\ranking^*(v) \le \weight(v,v') + \ranking^*(v')$ for all successors~$v' \in vE$. Furthermore, there is some successor~$\overline{v} \in vE$ with $\ranking^*(v) = \weight(v,\overline{v}) + \ranking^*(\overline{v})$. Finally, if $ \ranking^*(v) < \infty $, then $\overline{v}$ can be chosen such that it additionally satisfies $t_s(v) = t_s(\overline{v}) + 1 $. 

	\item\label{lemma-reachability-rankingproperties-playerone}
	If $v \in V_1 \setminus F$ then $\ranking^*(v) \ge \weight(v,v') + \ranking^*(v')$ for all successors~$v' \in vE$.
	 Furthermore, there is some successor~$\overline{v} \in vE$ with $\ranking^*(v) = \weight(v,\overline{v}) + \ranking^*(\overline{v})$. 

	\item\label{lemma-reachability-rankingproperties-playeroneequality}
	If $v \in V_1 \setminus F$ and $\overline{v} \in vE$ with $\ranking^*(v) = \ranking^*(\overline{v}) < \infty$, then $t_s(v) > t_s(\overline{v})$.
	
\end{enumerate}

Before we prove these, we state some basic facts about settling times, which follow immediately from their definition.

\begin{remark}
\label{remark-reachability-settlingtimes}
 Let $v \in V$.
\begin{enumerate}
	\item\label{remark-reachability-settlingtimes-finalstep} If $t_s(v) > 0 $ then $\ranking_{t_s(v)}(v) < \ranking_{t_s(v)-1}(v)$.
	\item\label{remark-reachability-settlingtimes-successors} $t_s(v) \le \max_{v' \in vE} t_s(v') +1 $.
\end{enumerate}	
\end{remark}

Now, we are ready to prove Lemma~\ref{lemma-reachability-rankingproperties}.

\begin{proof}
Items~\ref{lemma-reachability-rankingproperties-infty} and \ref{lemma-reachability-rankingproperties-F} are trivial by definition of $\lift$.

\ref{lemma-reachability-rankingproperties-playerzero}) We have
\begin{equation}
\ranking^*(v) = \min \set{ \ranking^*(v), \min_{v' \in vE} \weight(v,v') +\ranking^*(v') } \label{eq1}
\end{equation}
by $\lift(\ranking^*) = \ranking^*$. Hence, $\ranking^*(v) \le \min_{v' \in vE} \weight(v,v') +\ranking^*(v')$, which implies $\ranking^*(v) \le \weight(v,v') +\ranking^*(v')$ for every $v' \in vE$, i.e., we have proven the first claim.

Now, towards a contradiction, assume there is no successor~$\overline{v} \in vE$ with $\ranking^*(v) = \weight(v,\overline{v}) + \ranking^*(\overline{v})$. Then, $\ranking^*(v) < \weight(v,v') + \ranking^*(v')$  for all successors~$v' \in vE$, which,  due to monotonicity, implies
\begin{equation}
\ranking^*(v) < \weight(v,v') + \ranking_j(v')
\label{eq1.5}
\end{equation}
for all successors~$v' \in vE$ and all $j$.
Furthermore, due to the strict inequality in Equation~\ref{eq1.5}, we have $\ranking^*(v) < \infty$ and therefore $t_s(v) > 0$ due to Item~\ref{lemma-reachability-rankingproperties-infty}.
Now, we have
\[
\ranking^*(v) = \ranking_{t_s(v)}(v) = \min \set{ \ranking_{t_s(v)-1}(v), \min_{v' \in vE} \weight(v,v') + \ranking_{t_s(v)-1}(v') } = \min_{v' \in vE} \weight(v,v') + \ranking_{t_s(v)-1}(v') > \ranking^*(v),
\]
which yields the desired contradiction. Here, the third equality is due to  Remark~\ref{remark-reachability-settlingtimes}.\ref{remark-reachability-settlingtimes-finalstep} and the final inequality due to Equation~\ref{eq1.5} 
for $j = t_s(v)-1$.

Finally, consider the case $\ranking^*(v) < \infty$. Then, $\ranking_{t_s(v)}(v) < \ranking_{t_s(v)-1}(v)$ due to Item~\ref{lemma-reachability-rankingproperties-infty} and Remark~\ref{remark-reachability-settlingtimes}.\ref{remark-reachability-settlingtimes-finalstep} This implies 
\begin{equation}
\ranking^*(v) = \ranking_{t_s(v)}(v) =  \min \set{ \ranking_{t_s(v)-1}(v), \min_{v' \in vE} \weight(v,v') + \ranking_{t_s(v)-1}(v') } = \weight(v,\overline{v}) + \ranking_{t_s(v)-1}(\overline{v})
\label{eq2.5}
\end{equation}
for some $\overline{v} \in vE$. It remains to prove $t_s(v) = t_s(\overline{v}) + 1 $, as this allows us then to rewrite Equation~(\ref{eq2.5}) to $  \ranking^*(v) = 
 \weight(v,\overline{v}) + \ranking^*(\overline{v})
$. Then, $\overline{v}$ has the desired properties. 

First, towards a contradiction, assume we have $t_s(v) \le t_s(\overline{v}) $. Then,
\[
\ranking_{t_s(\overline{v})}(\overline{v}) < \ranking_{t_s(\overline{v})-1}(\overline{v}) \le \ranking_{t_s(v)-1}(\overline{v})
\]
where the first inequality is due to Remark~\ref{remark-reachability-settlingtimes}.\ref{remark-reachability-settlingtimes-finalstep} and the second one due to monotonicity of the $\ranking_j$ and our assumption~$t_s(v) \le t_s(\overline{v})$. Hence,
\[
\ranking_{t_s(\overline{v})+1}(v) = \min \set{ \ranking_{t_s(\overline{v})}(v), \min_{v' \in vE} \weight(v,v') + \ranking_{t_s(\overline{v})}(v') } \le 
\weight(v,\overline{v}) + \ranking_{t_s(\overline{v})}(\overline{v}) < \weight(v, \overline{v}) + \ranking_{t_s(v)-1}(\overline{v}) = \ranking^*(v),
\]
which yields the desired contradiction to the definition of $\ranking^*$. Here, the last equality is due to Equation~(\ref{eq2.5}). Hence, we have $t_s(v) > t_s(\overline{v}) $. 

Finally, towards a contradiction, assume we have $t_s(v) > t_s(\overline{v}) + 1$, i.e., $t_s(\overline{v}) \le t_s(v)-2$. Then,
\begin{multline*}
	\ranking_{t_s(v)-1}(v) = 
\min \set{ \ranking_{t_s(v)-2}(v), \min_{v' \in vE} \weight(v,v') + \ranking_{t_s(v)-2}(v') }  \le \weight(v,\overline{v}) + \ranking_{t_s(v)-2}(\overline{v}) =\\ \weight(v,\overline{v}) + \ranking_{t_s(v)-1}(\overline{v}) = \ranking^*(v),
\end{multline*}
where the first equality is due to $\overline{v}$ being settled at step~$t_s(v)-2$ and the last equality is due to Equation~(\ref{eq2.5}). However, $\ranking_{t_s(v)-1}(v)$ being at most $\ranking^*(v)$ contradicts the definition of $t_s(v)$. Hence, we must have $t_s(v) = t_s(\overline{v}) + 1 $ as required.

\ref{lemma-reachability-rankingproperties-playerone}) Let $t = t_s(v)$ be  the settling time of $v$. If $t_s(v) = 0$, then $\ranking^*(v) = \infty$ as claimed in Item~\ref{lemma-reachability-rankingproperties-infty}. Hence, $\ranking^*(v) \ge \weight(v,v') + \ranking^*(v')$ for every successor~$v' \in vE$, as $\infty$ is the maximal rank. 

Now, assume we have $t_s(v) > 0$, which implies $\ranking^*(v) = \ranking_{t_s(v)}(v) < \ranking_{t_s(v)-1}(v)$ due to Remark~\ref{remark-reachability-settlingtimes}.\ref{remark-reachability-settlingtimes-finalstep}. Then, 
\[
\ranking^*(v) = \ranking_{t_s(v)}(v) = \min \set{ \ranking_{t_s(v)-1}(v), \max_{v' \in vE} \weight(v,v') +\ranking_{t_s(v)-1}(v') } = \max_{v' \in vE} \weight(v,v') +\ranking_{t_s(v)-1}(v').\]
Hence, 
\begin{equation}
\label{eq2}
\ranking^*(v) \ge \weight(v,v') +\ranking_{t_s(v)-1}(v') \ge \weight(v,v') +\ranking^*(v')
\end{equation}
for every $v' \in vE$, 
where the last inequality is due to monotonicity of the $r_j$. Thus, we have proven the first claim.

Furthermore, we have $\ranking^*(v) = \min\set{ \ranking^*(v), \max_{v'\in vE} \weight(v,v') + \ranking^*(v') } $ due to $\lift(\ranking^*) = \ranking^*$, which implies $\ranking^*(v) \le \max_{v'\in vE} \weight(v,v') + \ranking^*(v') $. Let $\overline{v}$ realize the maximum, i.e., 
\begin{equation}\ranking^*(v) \le  \weight(v,\overline{v}) + \ranking^*(\overline{v}).\label{eq3}
\end{equation}
Combining Inequalities~(\ref{eq2}) and (\ref{eq3}) yields the desired equality~$\ranking^*(v) =  \weight(v,\overline{v}) + \ranking^*(\overline{v})$.

\ref{lemma-reachability-rankingproperties-playeroneequality}) 
Towards a contradiction, assume $t_s(v) \le t_s(\overline{v})$. 
Then, we have $
\ranking_{t_s(\overline{v})}(\overline{v}) < \ranking_{t_s(\overline{v})-1}(\overline{v})   	\le  \ranking_{t_s(v)-1}(\overline{v})
$
due Remark~\ref{remark-reachability-settlingtimes}.\ref{remark-reachability-settlingtimes-finalstep} and monotonicity. Thus, 
\begin{align*}
\ranking^*(v) = \ranking_{t_s(v)}(v) = {}& \min \set{ \ranking_{t_s(v)-1}(v) , \max_{v' \in vE} \weight(v,v') +\ranking_{t_s(v)-1}(v') } \\
= {}& \max_{v' \in vE} \weight(v,v') +\ranking_{t_s(v)-1}(v')\\
\ge {}& \weight(v,\overline{v}) +\ranking_{t_s(v)-1}(\overline{v})\\
\ge {} & \ranking_{t_s(v)-1}(\overline{v})\\
>{} & \ranking_{t_s(\overline{v})}(\overline{v})\\
={}& \ranking^*(\overline{v}),
\end{align*}
which yields the desired contradiction to $\ranking^*(v) = \ranking^*(\overline{v})$.
\end{proof}

\subsubsection{Proof of Lemma~\ref{lemma-reach-sigmaoptimal}}

As already mentioned in the main part, we prove that $\sigma$ is optimal in two steps. In Lemma~\ref{lemma-reachability-valsigma}, we show that $\ranking^*(v, \init(v))$ is an upper bound on the value of all plays starting in $v$ that are consistent with $\sigma$. 
Then, in Lemma~\ref{lemma-reachability-valtau}, we show that Player~$1$ has a strategy to enforce a play with value greater or equal to $\ranking^*(v, \init(v))$ from $v$. Thus, there is no better strategy for Player~$0$ than $\sigma$, i.e., $\sigma$ is optimal.

\begin{lemma}
\label{lemma-reachability-valsigma}
We have $\val_\game^R(\sigma, v) \le \ranking^*(v, \init(v))$ for every vertex~$v$ of $\arena$.
\end{lemma}

\begin{proof}
We can assume w.l.o.g.\ $\ranking^*(v, \init(v)) < \infty$, as the statement is vacuously true otherwise.
Let $\rho_0 \rho_{1} \rho_{2} \cdots$ be a play in $\arena$ starting in $v$ and consistent with $\sigma$. Furthermore, let $q_0q_{1}q_{2} \cdots$ be the unique run of $\aut$ on $\col(\rho)$, i.e., $q_{j} = \delta^*(\col(\rho_0 \cdots \rho_{j-1}))$. In particular, $q_1 = \init(v)$. By construction of $\arena \times \mem_{\aut} $, $\ext(\rho) = (\rho_0, q_1) (\rho_1, q_2) (\rho_2, q_3) \cdots$ is a play in $\arena \times \mem_{\aut} $ (note the shift between the indexes).
   
   First, we show by induction over $j $, that there either is a $ j' < j$ with $(\rho_{j'}, q_{j' + 1}) \in F$ or that we have 
\begin{equation}
\label{eq-5}
\ranking^*(\rho_0, q_1) \ge \weight( \rho_0 \cdots \rho_{j} ) + \ranking^*(\rho_{j}, q_{j+1}) .	
\end{equation}
The induction start~$j = 0$ is trivial, as the statement simplifies to $\ranking^*(\rho_0, q_1) \ge \ranking^*(\rho_0, q_1)$.
Thus, assume we have $j > 0$ and $(\rho_{j'}, q_{j'+1}) \notin F$ for every $j' < j$ (otherwise, we are done). Then, the induction hypothesis yields $\ranking^*(\rho_0, q_1) \ge \weight( \rho_0 \cdots \rho_{j-1} ) + \ranking^*(\rho_{j-1}, q_{j})$.

If it is Player~$0$'s turn at $\rho_{j-1}$, then $(\rho_{j}, q_{j+1})$ is an optimal successor for $(\rho_{j-1}, q_{j})$ by the definition of $\sigma$. Hence, $\ranking^*(\rho_{j-1}, q_{j}) = \weight(\rho_{j-1}, \rho_{j}) + \ranking^*(\rho_{j}, q_{j+1}) $. Similarly, if it is Player~$1$'s turn at $\rho_{j-1}$, then we have $\ranking^*(\rho_{j-1}, q_{j}) \geq \weight(\rho_{j-1}, \rho_{j}) + \ranking^*(\rho_{j}, q_{j+1}) $ due to Lemma~\ref{lemma-reachability-rankingproperties}.\ref{lemma-reachability-rankingproperties-playerone}. Hence, applying this to the induction hypothesis yields the desired result in both cases: We have
\begin{align*}
	\ranking^*(\rho_0, q_1) \ge{} &{} \weight( \rho_0 \cdots \rho_{j-1} ) + \ranking^*(\rho_{j-1}, q_{j}) \\
	\ge{} & {}\weight( \rho_0 \cdots \rho_{j-1} ) + \weight(\rho_{j-1}, \rho_{j}) + \ranking^*(\rho_{j}, q_{j+1})\\
	= {}&{}  \weight( \rho_0 \cdots  \rho_{j}) + \ranking^*(\rho_{j}, q_{j+1}).
\end{align*}

Next, we prove by reductio ad absurdum that there is a $j$ with $(\rho_{j}, q_{j+1}) \in F$. 
If there is no such $j$, then Items~\ref{lemma-reachability-rankingproperties-playerzero} and \ref{lemma-reachability-rankingproperties-playeroneequality} of  Lemma~\ref{lemma-reachability-rankingproperties} yield that for every $j$, either $\ranking^*(\rho_{j}, q_{j+1}) > \ranking^*(\rho_{j+1}, q_{j+2}) $ or both  $\ranking^*(\rho_{j}, q_{j+1}) = \ranking^*(\rho_{j+1}, q_{j+2}) $ and $t_s(\rho_{j}, q_{j+1}) > t_s(\rho_{j+1}, q_{j+2}) $. However, this yields an infinite decreasing chain in a lexicographic order which has no such chains, i.e., we have obtained the desired contradiction.

Hence, let $j$ be minimal with $(\rho_{j}, q_{j+1}) \in F$. By construction, $q_{j+1} = \delta^*(\col(\rho_0 \cdots \rho_{j}))$ is an accepting state of $\aut$. This in turn implies $\col(\rho_0 \cdots \rho_{j}) \in L(\aut)$ by construction. 

Applying the definition of $\val_\game^R$, minimality of $j$, $\ranking^*(\rho_{j}, q_{j+1}) = 0$ due to Lemma~\ref{lemma-reachability-rankingproperties}.\ref{lemma-reachability-rankingproperties-F}, and Equation~\ref{eq-5} yields
\[\val_\game^R(\rho_0 \rho_{1} \rho_{2} \cdots ) = \weight(\rho_0 \cdots \rho_{j} ) =  \weight( \rho_0 \cdots \rho_{j } ) + \ranking^*(\rho_{j}, q_{j+1}) \le \ranking^*(\rho_0, q_1)  .
\]
As this inequality holds for every play that is consistent with $\sigma$ and starts in $v = \rho_0$ and due to $q_1 = \init(v)$, we conclude $\val_\game^R(\sigma, v) \le \ranking^*(v, \init(v))$.
\end{proof}

After having proved the upper bound~$\val_\game^R(\sigma, v) \le \ranking^*(v, \init(v))$ we now show that it is optimal by constructing a strategy for Player~$1$ that enforces a value of at least~$\ranking^*(v, \init(v))$ when starting in $v$, against any strategy of Player~$1$. Hence, there cannot be a better strategy that $\sigma$. 

We again define the strategy for Player~$1$ in $\arena$ by first defining a positional strategy~$\tau'$ in $\arena \times \mem_{\aut}$. From $v \in V_1 \setminus F$, $\tau'$ moves to an optimal successor of $v$. For vertices in $V_1 \cap F$, $\tau'$ moves to an arbitrary successor of $v$. 
Finally, let $\tau$ be the unique finite-state strategy implemented by $\mem_{\aut}$ and $\nxt_{\tau'}$, the next-move function induced by $\tau'$.

\begin{lemma}
\label{lemma-reachability-valtau}
We have $\val_\game^R(\rho) \ge \ranking^*(v, \init(v))$ for every vertex~$v$ of $\arena$ and every play~$\rho$ starting in $v$ and being consistent with $\tau$.
\end{lemma}

\begin{proof}
Consider a play~$\rho = \rho_0 \rho_1 \rho_2 \cdots $ in $\arena$ starting in $v$ and consistent with $\tau$. Furthermore, let $\ext(\rho) = (\rho_0, q_1) (\rho_1,q_2) (\rho_2, q_3) \cdots$ be its extended play, i.e., $q_0 q_1 q_2 q_3 \cdots$ is the run of $\aut$ on $\col(\rho)$ and $q_1 = \init(v)$.

Here, we show by induction over $j $ that there either is a $ j' < j$ with $(\rho_{j'}, q_{j' + 1}) \in F$ or that we have 
\begin{equation}
\label{eq-6}
\ranking^*(\rho_0, q_1) \le \weight( \rho_0 \cdots \rho_{j} ) + \ranking^*(\rho_{j}, q_{j+1}) .	
\end{equation}
The induction start~$j = 0$ is trivial, as the statement simplifies to $\ranking^*(\rho_0, q_1) \le \ranking^*(\rho_0, q_1)$.
Thus, assume we have $j > 0$ and $(\rho_{j'}, q_{j'+1}) \notin F$ for every $j' < j$ (otherwise we are done). Then, the induction hypothesis yields $\ranking^*(\rho_0, q_1) \le \weight( \rho_0 \cdots \rho_{j-1} ) + \ranking^*(\rho_{j-1}, q_{j})$.

If it is Player~$1$'s turn at $\rho_{j-1}$, then $(\rho_{j}, q_{j+1})$ satisfies $\ranking^*(\rho_{j-1}, q_{j}) = \weight(\rho_{j-1}, \rho_{j}) + \ranking^*(\rho_{j}, q_{j+1}) $ due to being an optimal successor. Similarly, if it is Player~$0$'s turn at $\rho_{j-1}$, then we have $\ranking^*(\rho_{j-1}, q_{j}) \leq \weight(\rho_{j-1}, \rho_{j}) + \ranking^*(\rho_{j}, q_{j+1}) $ due to Lemma~\ref{lemma-reachability-rankingproperties}.\ref{lemma-reachability-rankingproperties-playerzero}. Hence, applying this inequality to the induction hypothesis yields the desired result:
\begin{align*}
	\ranking^*(\rho_0, q_1) {}\le & {}\weight( \rho_0 \cdots \rho_{j-1} ) + \ranking^*(\rho_{j-1}, q_{j}) \\
	\le {}&{} \weight( \rho_0 \cdots \rho_{j-1} ) + \weight(\rho_{j-1}, \rho_{j}) + \ranking^*(\rho_{j}, q_{j+1})\\
	= {}& {} \weight( \rho_0 \cdots \rho_{j}) + \ranking^*(\rho_{j}, q_{j+1}).
\end{align*}

Now, we consider two cases: if $\ranking^*(\rho_0, q_1) = \infty$, we show that there is no $j$ with $(\rho_j, q_{j+1}) \in F$. Then, by definition of $\aut$, $\rho$ has no prefix in $L(\aut)$ (here, we use $\epsilon \notin L(\aut)$). Hence, $\val_\game^R(\rho) = \infty \ge \ranking^*(\rho_0, q_1) $. 

Thus, towards a contradiction, assume there is a $j$ (which we assume w.l.o.g.\ to be minimal) with $(\rho_j, q_{j+1}) \in F$. Then, Equation~(\ref{eq-6}) and Lemma~\ref{lemma-reachability-rankingproperties}.\ref{lemma-reachability-rankingproperties-F} yield 
\[
\infty = \ranking^*(\rho_0, q_1) \le \weight( \rho_0 \cdots \rho_{j} ) + \ranking^*(\rho_{j}, q_{j+1}) = \weight( \rho_0 \cdots \rho_{j} ) < \infty,	
\] 
which yields the desired contradiction. 

Finally, consider the case where $\ranking^*(\rho_0, q_1) < \infty$. If there is no $j$ with $(\rho_j, q_{j+1}) \in F$, then we have, as in the previous case, $\val_\game^R(\rho) = \infty > \ranking^*(\rho_0, q_1) $.  On the other hand, if there is a $j$ with $(\rho_j, q_{j+1}) \in F$ (which we again pick to be minimal with this property), then  $\rho_0 \cdots \rho_j \in L(\aut)$ and we have 
\[
\val_\game^R(\rho) = \weight(\rho_0 \cdots \rho_j) = \weight(\rho_0 \cdots \rho_j) + \ranking^*(\rho_j, q_{j+1}) \ge \ranking^*(\rho_0, q_1),
\]
again due to Equation~(\ref{eq-6}), minimality of $j$, and Lemma~\ref{lemma-reachability-rankingproperties}.\ref{lemma-reachability-rankingproperties-F}.
\end{proof}

\subsubsection{Proof of Lemma~\ref{lemma-reachability-termination}}

We need to prove $\ranking^* = \ranking_{\size{\arena} \cdot \size{\aut}+1}$.

\begin{proof}
It suffices to prove that the settling time of every vertex~$v$ in $\arena \times \mem_{\aut}$ is smaller than $\size{\arena} \cdot \size{\aut}+1$. This is in particular true for every $v$ with $\ranking^*(v) =\infty$, as those have settling time $0$.

Recall that we defined a positional strategy~$\sigma'$ for Player~$0$ in $\arena \times \mem_{\aut}$ that always moves from a vertex~$v \in V_0 \setminus F$ to an optimal successor. 
Now, given a vertex~$v$ of $\arena \times \mem_{\aut}$ with $ \ranking^*(v) < \infty $, let $d(v)$ be the length of the longest play prefix starting in $v$, consistent with $\sigma'$, and never visiting $F$. Using the descending-chain argument from the proof of Lemma~\ref{lemma-reachability-valsigma} shows that $d(v)$ is well-defined and bounded by $\size{\arena} \cdot \size{\aut}$: if there was a longer path, then it has a cycle and therefore we obtain an infinite play that is consistent with $\sigma'$, but never visits $F$. From this, one can again construct an infinite descending chain in the lexicographic order on the ranks and on the settling times. 

Hence, it suffices to prove by induction over $d(v)$ that we have $t_s(v) \le d(v)+1$ for every $v$ with $\ranking^*(v) < \infty$. The induction start is simple, as $d(v) = 0$ implies $v \in F$. Hence, $v$ has settling time~$1$ due to Lemma~\ref{lemma-reachability-rankingproperties}.\ref{lemma-reachability-rankingproperties-F}.

Now, consider a vertex~$v$ with $d(v) > 0$. First, we assume $v \in V_0$. Let $\sigma'(v) = v'$, which is an optimal successor of $v$ by  construction of $\sigma'$. As every play prefix that contributes to $d(v)$ visits $v'$ as second vertex, we conclude $d(v') = d(v)-1$, i.e., 
the induction hypothesis is applicable to $v'$. Hence, we have~$t_s(v) = t_s(v') + 1 \le d(v') + 1 +1 = d(v) +1$ due to Lemma~\ref{lemma-reachability-rankingproperties}.\ref{lemma-reachability-rankingproperties-playerzero} and the induction hypothesis.

Finally, consider the case $v \in V_1$. Let $v' \in vE$ be a successor of $v$. We have $d(v') \le d(v) - 1 $, i.e., the induction hypothesis is applicable. Hence, we have $t_s(v') \le d(v') + 1 \le d(v) $. Thus, the settling time of the successors of $v$ is at most $d(v)$. Hence, the settling time of $v$ is at most $d(v)+1$ due to Remark~\ref{remark-reachability-settlingtimes}.\ref{remark-reachability-settlingtimes-successors}. Thus, we have $t_s(v) \le d(v)+1$ as required. 
\end{proof}


\subsection{Proofs Omitted in Subsection~\ref{subsection-limitalgo}}

\subsubsection{Proof of Lemma~\ref{lemma-limit-termination-naive}}

We need to show that $\ranking_j(v)< \infty$ implies $\ranking_j(v) \le (\size{\arena} \cdot \size{\aut} +1) \cdot W$.

\begin{proof}
By induction over $j$. The induction start is trivial, as $\ranking_0$ maps every vertex to $0$. Now, assume the bound holds for $j \ge 0$. Consider the computation of $\ranking_{j+1}$: by induction hypothesis, every $\mathfrak{r}_h$, which is equal to $\ranking_{j}(v)$ for some $v \in V$, is bounded by $(\size{\arena} \cdot \size{\aut} +1) \cdot W$ or infinite. Furthermore, the ranks assigned by each $\ranking_{h}'$ are bounded by $(\size{\arena} \cdot \size{\aut}) \cdot W$ (Corollary~\ref{coro-reachability-upperboundvalues}) or infinite. Thus, the ranks assigned by each $\ranking_{h}''$ are bounded by $(\size{\arena} \cdot \size{\aut} +1) \cdot W$ or infinite. Hence, all finite values that can contribute to the ranks of $\ranking_{j+1}$ are bounded by $(\size{\arena} \cdot \size{\aut} +1) \cdot W$.
\end{proof}

\subsubsection{Proof of Lemma~\ref{lemma-limit-valsigma}}

We need to prove $\val_{\game}(\sigma, v) \le \ranking^*(v, \init(v)) $ for every $v$ in $\arena$.	

\begin{proof}
We assume $\ranking^*(v, \init(v)) < \infty$, as the result is trivial otherwise. 

Let $\rho= \rho_0 \rho_1 \rho_2 \cdots$ be a play consistent with $\sigma$ starting in $v$ and let $q_0 q_1 q_2 \cdots$ be the run of $\aut$ on $\rho$. By construction, $\rho' = (\rho_0, q_1)(\rho_1, q_2)(\rho_2, q_3) \cdots$ is a play that is consistent with $\sigma'$ and starts in $(v, \init(v))$. Furthermore, let $h_0 h_1 h_2 \cdots $ be the sequence of memory elements assumed by $\mem'$ during $\rho'$. In particular, $h_0 = \init'(v, \init(v))$. 

We say that a position~$j$ is a change-point, if $j =0$ or if $j > 0$ and $(\rho_j,q_{j+1}) \in F_{h_{j-1}} \subseteq F$. By construction of $\mem'$, we have $h_{j} = \init'(\rho_j, q_{j+1}) = h(\rho_j, q_{j+1})$ for every change-point~$j$. 
 
Let $j$ be a change-point. We show that there is a next change-point~$j' > j$ with $\weight((\rho_j, q_{j+1} )\cdots (\rho_{j'}, q_{j'+1})) \le \ranking^*(v, \init(v))$, i.e., the accumulated weight between change-points is bounded by $\ranking^*(v, \init(v))$. Recall that the weight of an infix of $\rho$ coincides with the weight of the corresponding infix of $\rho'$. Furthermore, $(\rho_j, q_{j+1}) \in F$ for every non-zero change-point implies that $q_{j+1} = \delta^*(\col(\rho_0 \cdots \rho_j))$ is an accepting state of $\aut$. Altogether, this implies $\val_\game(\rho) \le \ranking^*(v, \init(v))$. 

We prove this claim while showing inductively that the sequence~$(\ranking^*(\rho_j, q_{j+1}))_j$, where $j$ ranges over the change-points, is weakly decreasing. As the induction start and the induction step are similar, we deal with both at the same time, i.e., let $j$ be an arbitrary change-point. It suffices to show that there is a next change-point~$j' > j$ with $\weight((\rho_j, q_{j+1} )\cdots (\rho_{j'}, q_{j'+1})) \le \ranking^*(v, \init(v))$ and $\ranking^*(\rho_j, q_{j+1})\ge  \ranking^*(\rho_{j'}, q_{j'+1})$.  To this end, we need to distinguish two cases:

First, assume we have $(\rho_j, q_{j+1}) \in F_{h_{j}}$. Then, by definition of $\sigma'$ and by the definition of $\ranking_{h_{j}} = \complete_{F_{h_{j}}}(\ranking_{h_{j}}')$ we have 
\begin{equation}
\label{eqlimit1}
\ranking_{h_{j}} (\rho_j, q_{j+1}) \ge \weight((\rho_j, q_{j+1}), (\rho_{j+1}, q_{j+2})) + \ranking'_{h_{j}} ((\rho_{j+1}, q_{j+2})). 
\end{equation}
Now, we need to consider two subcases. First, if $(\rho_{j+1}, q_{j+2}) \in F_{h_j}$, which implies that $j+1$ is a change-point, then we have 
\[\weight( (\rho_{j}, q_{j+1}) (\rho_{j+1}, q_{j+2}) ) \le \ranking_{h_{j}} (\rho_j, q_{j+1}) \le \ranking^*( \rho_j, q_{j+1} ) \le \ranking^*( \rho_0, q_1 ) .\]
Here, the first inequality is Equation~(\ref{eqlimit1}) and $\ranking'_{h_{j}} (\rho_{j+1}, q_{j+2})$ being non-negative. The second inequality is due to the induction hypothesis. Thus, we have proven that a next change-point is reached with the required accumulated weight. Furthermore,
\begin{equation}
\label{eqFlimit}
\ranking^*(\rho_{j+1}, q_{j+2}) \le \mathfrak{r}_{h_j} \le \ranking^*(\rho_{j}, q_{j+1}),
\end{equation}
 where the first inequality is due to $(\rho_{j+1}, q_{j+2}) \in F_{h_j}$ and the second one is due to $ \ranking^*(\rho_j, q_{j+1}) \le \mathfrak{r}_{h(\rho_j, q_{j+1})} = \mathfrak{r}_{h_j} $. Here, we use the fact that $j$ is a change-point, which means that $h_j$ is updated to $h(\rho_j, q_{j+1})$.

In the second subcase, we have $(\rho_{j+1}, q_{j+2}) \notin F_{h_j}$. Then, $\sigma'$ behaves by construction like an optimal strategy to reach $F_{h_j}$ from $(\rho_{j+1}, q_{j+2})$ with accumulated weight at most $\ranking_{h_j}'(\rho_{j+1}, q_{j+2})$. Say $F_{h_j}$ is reached at position~$j' > j+1$. Then, $j' > j$ is a change-point and we have 
\begin{align*}
\weight( (\rho_{j}, q_{j+1}) \cdots (\rho_{j'}, q_{j'+1}) ) = & \weight((\rho_{j}, q_{j+1})(\rho_{j+1}, q_{j+2})) + \weight( (\rho_{j+1}, q_{j+2}) \cdots (\rho_{j'}, q_{j'+1}) ) \le \\
&\weight((\rho_{j}, q_{j+1})) + \ranking_{h_j}'(\rho_{j+1}, q_{j+2}))  \le \ranking_{h_{j}} (\rho_j, q_{j+1})
\end{align*}
 due to Equation~(\ref{eqlimit1}). Finally, $\ranking^*((\rho_{j+1}, q_{j+2})) \ge \ranking^*((\rho_{j'}, q_{j'}))$ is again implied by $(\rho_{j'}, q_{j'+1}) \in F_{h_j}$ (cf.\ Equation~(\ref{eqFlimit})).
 Thus, the change-point~$j'$ has all the desired properties.
 
 To conclude consider the case $(\rho_j, q_{j+1}) \notin F_{h_{j}}$. Here, we proceed analogously to the second subcase above: Now, $\sigma'$ behaves like an optimal strategy to reach $F_{h_j}$ from $(\rho_{j}, q_{j+1})$ with accumulated weight at most $\ranking_{h_j}'(\rho_{j}, q_{j+1})$. Hence, a changepoint~$j$ is reached with accumulated weight at most  $\ranking_{h_j}'(\rho_{j}, q_{j+1}) \le \ranking^*(\rho_{j}, q_{j+1}) \le \ranking^*(\rho_{0}, q_{1})$ as desired. Furthermore, $\ranking^*((\rho_{j+1}, q_{j+2})) \ge \ranking^*((\rho_{j'}, q_{j'}))$ holds due to the same arguments as above.
\end{proof}

\subsubsection{Proof of Lemma~\ref{lemma-limit-valtau}}

We need to prove $\val_{\game}(\tau, v) \ge \ranking^*(v, \init(v)) $ for every $v$ in $\arena$.	

Before we present the proof, let us note that adding goal vertices can only decrease the weight of reaching a goal vertex, an observation which is formalized in the next remark. It follows from a straightforward induction over the inflationary computation of the fixed point as given in Subsection~\ref{subsection-reachalgo}.

\begin{remark}
\label{remark-reachability-rankingmonotonicity-changingF}
Fix $F' \subseteq F''$ and let $\ranking'$ and $\ranking''$ be the least fixed points of $\lift_{F'}$ and $\lift_{F''}$, respectively. Then, $\ranking' \sqsubseteq \ranking''$, i.e., $\ranking'(v) \ge \ranking''(v)$ for every $ v \in V$. 
\end{remark}

Now, we are ready to prove Lemma~\ref{lemma-limit-valtau}.

\begin{proof} 
Let $\rho= \rho_0 \rho_1 \rho_2 \cdots$ be a play consistent with $\tau$ starting in $v$ and let $q_0 q_1 q_2 \cdots$ be the run of $\aut$ on $\rho$. By construction, $\rho' = (\rho_0, q_1)(\rho_1, q_2)(\rho_2, q_3) \cdots$ is a play that is consistent with $\tau'$ and starts in $(v, \init(v))$. Furthermore, let $m_0 m_1 m_2 \cdots $ be the sequence of memory elements assumed by $\mem'$ during $\rho'$. In particular, $m_0 = \init'(v, \init(v))$. 

We say that a position~$j$ is a change-point, if $j =0$ or if $j > 0$ and $(\rho_j,q_{j+1}) \in F $. By construction of $\mem'$, we have $m_{j} = \init'(\rho_j, q_{j+1}) = (\rho_j, q_{j+1}) $ for every change-point~$j$. Note that if $\rho$ has only finitely many change-points, then only finitely many prefixes of $\rho$ are accepted by $\aut$. Hence, we have $\val_\game(\rho) = \infty \ge \ranking^*(v, \init(v))$. For this reason, we only consider the case were $\rho$ has infinitely many change-points.

Now, let $j$ be a change-point and $j' > j$ be the next change-point. We show that either one of the following three possibilities holds:
\begin{enumerate}
	\item $\ranking^*(\rho_j, q_{j+1}) < \ranking^*(\rho_{j'}, q_{j'+1})$, i.e., the rank strictly increases.
	\item $\ranking^*(\rho_j, q_{j+1}) = \ranking^*(\rho_{j'}, q_{j'+1})$ and $t_s(\rho_j, q_{j+1}) > t_s(\rho_{j'}, q_{j'+1})$, i.e., the rank is constant, but the settling time strictly decreases.
	\item $\weight( (\rho_j, q_{j+1}) \cdots (\rho_{j'}, q_{j'+1}) ) \ge \ranking^*(\rho_j, q_{j+1})$, i.e., we have an infix of weight at least $\ranking^*(\rho_j, q_{j+1})$.
\end{enumerate}
Applying the other two possibilities inductively, we obtain that the first infix as in the third possibility has at least weight~$\ranking^*(\rho_0, q_{1})$, which implies $\val_\game(\rho) \ge \ranking^*(\rho_0, q_{1})$ as required. Thus, it remains to show that the third possibility eventually holds. But this is straightforward, since in the other two cases, either the rank strictly increases or the rank is constant and the settling time decreases. This cannot happen infinitely often, as there are only finitely many possible ranks and the settling times are non-negative.

To complete the proof, we show that either one of the three possibilities above holds for any pair of adjacent change-points~$j$ and $j'$. To this end, we consider the different types $(\rho_j, q_{j+1})$ can have.

If $(\rho_j, q_{j+1})$ has type zero, i.e., if $\ranking^*(\rho_j, q_{j+1}) = 0$ , then we have 
	\[
	\weight( (\rho_j, q_{j+1}) \cdots (\rho_{j'}, q_{j'+1}) ) \ge \ranking^*(\rho_j, q_{j+1})
	\]
	as the accumulated weight is always non-negative. Hence, the third possibility holds.
	
	 If $(\rho_j, q_{j+1})$ has type one, then $(\rho_j, q_{j+1}) \cdots (\rho_{j'}, q_{j'+1})$ is consistent with an optimal strategy with respect to $\ranking_{h(\rho_j, q_{j+1})}''$ and we consider two subcases.
	
If $(\rho_{j'}, q_{j'+1}) \in F_{h(\rho_j, q_{j+1})}$, then we have 
	\[
	\weight( (\rho_j, q_{j+1}) \cdots (\rho_{j'}, q_{j'+1}) ) \ge \ranking_{ h(\rho_j, q_{j+1}) }''(\rho_j, q_{j+1}) = \ranking^*(\rho_j, q_{j+1})
	\]
by optimality of the strategy and by definition of $h(\rho_j, q_{j+1})$.
Hence, the third possibility holds.

On the other hand, assume we have $(\rho_{j'}, q_{j'+1}) \in F \setminus F_{h(\rho_j, q_{j+1})}$. Then, let $h$ be minimal with $(\rho_{j'}, q_{j'+1}) \in F_h$, which implies 
\[	\mathfrak{r}_h = \ranking_{t_s(\rho_j, q_{j+1})-1}(\rho_{j'}, q_{j'+1}) \le \ranking^*(\rho_{j'}, q_{j'+1}). 
\]
Furthermore, due to Remark~\ref{remark-reachability-rankingmonotonicity-changingF} and $h > h(\rho_{j}, q_{j+1})$, we have $\ranking_h''(\rho_{j}, q_{j+1}) \le \ranking_{h(\rho_{j}, q_{j+1})}''(\rho_{j}, q_{j+1})$. Note that then
\[
\mathfrak{r}_h \leq \ranking_{h(\rho_{j}, q_{j+1})}(\rho_{j}, q_{j+1})
\]
contradicts the definition of $h(\rho_{j}, q_{j+1})$: then both $\ranking_h''(\rho_{j}, q_{j+1})$ and $\mathfrak{r}_h$ are at most $\ranking_{h(\rho_{j}, q_{j+1})}''(\rho_{j}, q_{j+1})$, i.e., $h(\rho_{j}, q_{j+1})$ is not the maximal index in the minimization.
Hence, altogether we have
\[
\ranking^*(\rho_{j}, q_{j+1}) = \ranking_{h(\rho_{j}, q_{j+1})}(\rho_{j}, q_{j+1}) < \mathfrak{r}_h \le \ranking^*(\rho_{j'}, q_{j'+1}),
\]
i.e., the first possibility holds.
	
	 If $(\rho_j, q_{j+1})$ has type two and we have $h(\rho_j, q_{j+1}) = 1$ then
	\[
	\ranking^*(\rho_{j'}, q_{j'+1}) \ge \ranking_{t_s(\rho_j, q_{j+1})-1}(\rho_{j'}, q_{j'+1}) = \mathfrak{r}_h 
	\]  
	for some $h$. As the $\mathfrak{r}_h$ are strictly increasing, we obtain $\ranking^*(\rho_j, q_{j+1}) \le \ranking^*(\rho_{j'}, q_{j'+1})$.
	
	Furthermore, $\ranking^*(\rho_{j'}, q_{j'+1}) = \mathfrak{r}_h$ implies that the settling time of $(\rho_{j'}, q_{j'+1})$ is strictly smaller than that of $(\rho_{j}, q_{j+1})$. Altogether, either the first or the second possibility holds.
	
	Finally, assume $(\rho_j, q_{j+1})$ has type two and we have $h(\rho_j, q_{j+1}) > 1$. Then, $(\rho_j, q_{j+1}) \cdots (\rho_{j'}, q_{j'+1})$ is consistent with an optimal strategy with respect to $\ranking_{h(\rho_j, q_{j+1})-1}''$. 
	If $(\rho_{j'}, q_{j'+1})$ is in $F_{h(\rho_j, q_{j+1})-1}$, then we have 
	\[
	\weight( (\rho_j, q_{j+1}) \cdots (\rho_{j'}, q_{j'+1}) ) \ge \ranking_{ h(\rho_j, q_{j+1})-1 }(\rho_j, q_{j+1}) > \mathfrak{r}_{h(\rho_j, q_{j+1})} = \ranking^*(\rho_j, q_{j+1}),
	\]
i.e., the third possibility holds. Here, the inequality~$\ranking_{ h(\rho_j, q_{j+1})-1 }(\rho_j, q_{j+1}) > \mathfrak{r}_{h(\rho_j, q_{j+1})} $ is by construction: we have $\mathfrak{r}_{h(\rho_j, q_{j+1})-1} < \mathfrak{r}_{h(\rho_j, q_{j+1})}$, but 
 \[
 \max\set{ \ranking_{h(\rho_j, q_{j+1})-1}'', \mathfrak{r}_{h(\rho_j, q_{j+1})-1} } > \max\set{ \ranking_{h(\rho_j, q_{j+1})}'', \mathfrak{r}_{h(\rho_j, q_{j+1})} } = \mathfrak{r}_{h(\rho_j, q_{j+1})},
  \]
 which yields the desired inequality.
 
On the other hand, assume $(\rho_{j'}, q_{j'+1})$ is in $F \setminus F_{h(\rho_j, q_{j+1})-1}$. Then,
\[
\ranking^*(\rho_{j'}, q_{j'+1}) \ge \mathfrak{r}_{h(\rho_j, q_{j+1})} = \ranking^*(\rho_{j}, q_{j+1}).
\]
Now, if we have $\ranking^*(\rho_{j'}, q_{j'+1}) = \ranking^*(\rho_{j}, q_{j+1})$ then also   $\ranking^*(\rho_{j'}, q_{j'+1}) = \mathfrak{r}_{h(\rho_j, q_{j+1})}$, which implies that the settling time of $(\rho_{j'}, q_{j'+1})$ is strictly smaller than the one of $\rho_{j}, q_{j+1}$, which is used to compute $\mathfrak{r}_{h(\rho_j, q_{j+1})}$. Hence, either the first or second possibility holds.
 \end{proof}

\subsubsection{Proof of Lemma~\ref{lemma-limit-termination}}

We need to show $\ranking^* = \ranking_{\size{F}+1}$.

\begin{proof}
We define for every play~$\rho$ a distance~$d(\rho)$ as follows: if $\rho$ visits $F$ finitely often, then $d(\rho)$ is defined to be the number of such visits. On the other hand, if $\rho$ visits $F$ infinitely often, then we define $d(v)$ to be the number of visits to $F$ before the first infix~$\rho_j \cdots \rho_{j'}$ of $\rho$ satisfying $ \rho_{j''} \notin F $ for all $j < j'' <j'$ and $\weight(\rho_j \cdots \rho_{j'}) \ge \ranking^*(\rho_0)$, where $\rho_0$ is the first vertex of $\rho$. We call such an infix a witness for $\rho$.
Now, we define $d(v) = \max_{\rho} d(\rho)$ where $\rho$ ranges over all plays that start in $v$ and are consistent with the strategy~$\tau$ for Player~$1$ defined above.

First, we show that $d(v)$ is at most $\size{F}$ for every $v$. Towards a contradiction, assume there is a play~$\rho$ contributing to $d(v)$ with $d(\rho) > \size{F}$. Recall that we termed visits to $F$ change-points in the proof of Lemma~\ref{lemma-limit-valtau} and showed that between any two adjacent changepoints (the first position is a change-point, too) either $\ranking^*$ strictly decreases, $\ranking^*$ stays constant and the settling time strictly decreases, or the infix between the change-points is a witness for $\rho$. The distance~$d(\rho)$ being greater than $F$ yields a repetition of a vertex in $F$ before a witness for $\rho$ appears (if it does at all). Using the lexicographic order induced by the ordering on the ranks and the settling times, we obtain the desired contradiction. 

To conclude the proof, we show $t_s(v) \le d(v)+1$ for all vertices~$v$ by induction over~$d(v)$. Here, we often use the trivial observation that $\ranking_j(v) \ge \ranking^*(v)$ implies $t_s(v) \le j$. 

First, let $d(v) = 0$ and consider first the case where $\ranking^*(v) = \infty$. Then, as there are no witnesses for a play~$\rho$ of value~$\infty$, every play starting in $v$ that is consistent with $\tau$ never visits $F$. This strategy witnesses that we have $\ranking_1(v) = \infty$, as it prevents Player~$0$ from reaching $F$. Now, consider the case $d(v) = 0$ and $\ranking^*(v) < \infty$. Then, every play starting in $v$ that is consistent with $\tau$ either does not visit $F$ at all or starts with a witnessing infix. Again, such a strategy witnesses $\ranking_1(v) \ge \ranking^*(v)$. In both cases, we have $t_s(v) = 1$.

Now, consider the case~$d(v) > 0$ and let $F' \subseteq F$ be the set of vertices $v' \in F$ such that there is a play starting in $v$, consistent with $\tau$, and where $v'$ is the first vertex in $F$ that is visited by $\rho$. By definition of $d$, we obtain $d(v') < d(v)$ for every $v' \in F'$. Hence, the induction hypothesis is applicable and we obtain $t_s(v') \le d(v') +1 $ for all $v' \in F'$. 

Let $t = \max_{v' \in F'} t_s(v')$. We prove below that we have $\ranking_{t+1}(v) \ge \ranking^*(v)$, which yields the desired result due to $t+1 \le d(v)+1$.

Thus, consider the computation of $\ranking_{t+1}(v) = \min_h \max\set{ \ranking_t(v) , \ranking_h''(v), \mathfrak{r}_h }$. If $\ranking_t(v) \ge \ranking^*(v)$, then we are done. Hence, it remains to consider the case with  $\ranking_t(v) < \ranking^*(v)$, which yields $\ranking_{t+1}(v) = \min_h \max\set{  \ranking_h''(v), \mathfrak{r}_h }$. Thus, we fix some $h$ and show that $\ranking_h''(v) < \ranking^*(v)$ implies $\mathfrak{r}_h \ge \ranking^*(v)$, which yields the desired result.

If $\ranking_h''(v) < \ranking^*(v)$ then Player~$0$ has a strategy~$\sigma_0$ to reach $F_h$ from $v$ with accumulated weight~$\ranking_h''(v)$ or smaller. Let $v' \in F_h$ be the vertex that is reached when Player~$1$ uses the strategy~$\tau$ against $\sigma_0$ starting in $v$. To conclude the proof, we assume towards a contradiction that we have $\mathfrak{r}_h < \ranking^*(v)$. Now, $v' \in F$ implies $\ranking_t(v) \le \mathfrak{r}_h < \ranking^*(v)$. Furthermore, as the settling time of $v'$ is at most $t$, we also have $\ranking^*(v') = \ranking_t(v') < \ranking^*(v)$. Hence, by Lemma~\ref{lemma-limit-valsigma}, Player~$0$ has a strategy~$\sigma$ that guarantees that the value of a play starting in $v'$ is at most $\ranking^*(v')$, and therefore strictly smaller than $\ranking^*(v)$.

Now, consider the following play~$\rho$ starting in $v$: Player~$1$ uses $\tau$ and Player~$0$ uses $\sigma_0$ until $v'$ is reached. Then, she uses $\sigma$ (discarding the history before the first visit of $v'$). As the play starts in $v$ and is consistent with $\tau$, we conclude $\val_\game(\rho) \ge \ranking^*(v)$. On the other hand, we claim $\val_\game(\rho) < \ranking^*(v)$, which yields the desired contradiction. The weight of the prefix of $\rho$ up to the first occurrence of $v'$ is strictly smaller than $\ranking^*(v)$ (due to the prefix being consistent with $\sigma_0$) and the value of the remaining suffix is also strictly smaller than $\ranking^*(v)$ (due to the suffix being consistent with $\sigma$). Hence, the value of the complete play is also strictly smaller than $\ranking^*(v)$.
\end{proof}

\subsection{Proofs Omitted in Section~\ref{section-pushdown}}

\subsubsection{Proof of Lemma~\ref{lemma-refinement}}

We need to show $W_0(\game) = \set{v \mid \val_\game(\sigma,v) < \infty}$ and $W_1(\game) = \set{v \mid \val_\game(\sigma,v) = \infty}$.

\begin{proof}
Assume $\val_\game(\sigma,v) < \infty$. Then, by definition of $\val_\game$ and Remark~\ref{remark-limit-valuegeneralizeswinning-play}, every play~$\rho$ that starts in $v$ and is consistent with $\sigma$ satisfies $c(\rho) \in \lim(K)$.  
Hence, $\sigma$ is a winning strategy from $v$, i.e., $v \in W_0(\game)$.

Now, assume $\val_\game(\sigma,v) = \infty$. Towards a contradiction, assume we have $v \notin W_1(\game)$. As $\lim(K)$ is Borel, $\game$ is determined~\cite{Martin75}, i.e., $v \in W_0(\game)$. Furthermore, as $\lim(K)$ is $\omega$-regular, Player~$0$ has a finite-state winning strategy for $\game$ from $v$~\cite{BuechiLandweber69}. 

Standard pumping arguments show that every play that starts in $v$ and is consistent with a finite-state winning strategy with $n$ memory elements has a value of at most $n\cdot\size{V}\cdot W$, where $V$ is the set of vertices of $\game$ and $W$ is the largest weight of $\game$. This contradicts $\val_\game(\sigma,v) = \infty$, i.e., we have derived the desired contradiction to $v \notin W_1(\game)$.

Thus, we have shown $\set{v \mid \val_\game(\sigma,v) < \infty} \subseteq W_0(\game)$ and $\set{v \mid \val_\game(\sigma,v) = \infty} \subseteq W_1(\game)$. As $\set{v \mid \val_\game(\sigma,v) < \infty} $ and $\set{v \mid \val_\game(\sigma,v) = \infty} $ partition the set of vertices, we obtain the desired result. 
\end{proof}